\documentclass{article}[11pt]
\pdfoutput=1
\usepackage[pdftex]{graphicx}
\usepackage{wrapfig,color,array}
\usepackage{algorithm2e}
\usepackage[noend]{algorithmic}
\usepackage[papersize={8.5in,11in},top=1.25in,left=1in,right=1in,bottom=1.25in]{geometry}
\usepackage{amsmath}
\usepackage{amsfonts}
\usepackage{multirow}
\usepackage{booktabs}
\usepackage{enumerate}
\usepackage{cite}

\def\math#1{$#1$}

\def\mand#1{$$#1$$}

\def\frac#1#2{{#1\over #2}}

\def\mld#1{\begin{equation}
#1
\end{equation}}

\def\choose#1#2{\left({{#1}\atop{#2}}\right)}

\def\a{{\mathbf a}}

\def\r#1{{(\ref{#1})}}

\def\dotfil{\leaders\hbox to 1.5mm{.}\hfill}
\newcounter{rmnum}
\def\RN#1{\setcounter{rmnum}{#1}\uppercase\expandafter{\romannumeral\value{rmnum}}}
\def\rn#1{\setcounter{rmnum}{#1}\expandafter{\romannumeral\value{rmnum}}}

\newtheorem{theorem}{Theorem}
\newtheorem{lemma}{Lemma}
\newtheorem{proposition}{Proposition}
\def\a{\alpha}
\def\g{\Gamma}
\def\p{\psi}
\def\G{\mathbb{G}}
\def\S{\mathbb{S}}
\newcommand{\remove}[1]{}

\date{}

\begin{document}
\title{Seeding Influential Nodes in Non-Submodular Models of Information Diffusion}
\author{%
Elliot Anshelevich\\
CS Department, RPI.\\
{\sf eanshel@cs.rpi.edu}
\and
Ameya Hate\\
CS Department, RPI.\\
{\sf hatea2@rpi.edu}
\and
Malik Magdon-Ismail\\
CS Department, RPI.\\
{\sf magdon@cs.rpi.edu}
}

\maketitle

\begin{abstract}%
We consider the model of information diffusion in social networks from
\cite{Hui2010a} which incorporates
trust (weighted links) between actors, and allows actors to actively
participate in the spreading process, specifically through the ability
to query friends for additional information. This model
captures how social agents transmit and act upon information more
realistically as compared to the simpler threshold and cascade models.
However, it is more difficult to analyze, in particular with respect
to seeding strategies.
We present efficient, scalable algorithms for determining good seed sets
-- initial nodes to inject with the information.
Our general approach is to reduce our model to a class of simpler
models for which provably good sets can be constructed. By tuning
this class of simpler models, we obtain a good seed set for the original
more complex model. We call this the \emph{projected greedy approach}
because you `project' your model onto a class of simpler models where
a greedy seed set selection is near-optimal. We demonstrate the
effectiveness of our seeding strategy on synthetic graphs as well as a
realistic San Diego evacuation network constructed during the 2007 fires.
\end{abstract}

\section{Introduction}

Networks (social, computer, and physical) are replete with the flow of
information, ideas, innovations, etc., and it is these
flows which affect the way people think, act, and bind together in a
society. Ideally, important messages should disseminate quickly and reach
the people who need to take action, and the diffusion of malicious
gossip should, if possible, be
terminated. Since diffusion of both useful and malicious items is at the
core of our society, it is vital to understand the mechanisms of diffusion
through dynamic networks. A network can change as a result of the diffusion.
For example, an
evacuation message may not reach its intended audience because certain
important
people (critical conduits of information)
left the community before the diffusion completed.
In this paper, we study how to optimize a diffusion in
realistic, large scale (multi-million node) complex networks; in
particular, how to select those actors to be initially seeded
with information so as to maximize the ultimate number of actors receiving the
information \emph{and} acting upon that information. Of particular interest
to us is the diffusion of high-value actionable information -- information
which is asking the user to take some action --  in particular
diffusion of an evacuation warning. This will be the context of our discussion,
however our methods are general.

\begin{wrapfigure}[14]{r}{3in}
\vspace*{-0.25in}
\begin{center}
\rotatebox{-90}{
\resizebox{!}{3in}{\rotatebox{90}{\includegraphics*{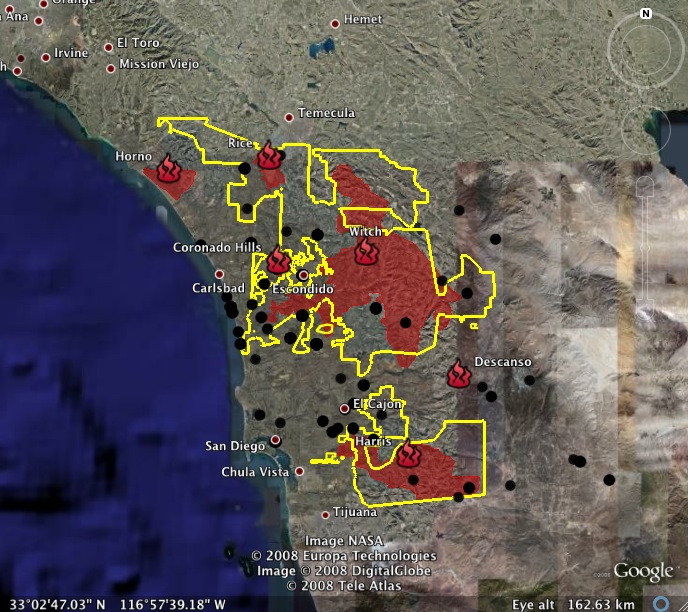}}}
}
\caption{\small 2007 San Diego fires evacuation area.\label{fig:SDfire}}
\end{center}
\end{wrapfigure}
The figure on the right shows the area which was affected by fires in 2007 (shaded
red),
the
mandatory evacuation area (yellow boundary),
and regions of unwarned people (black
dots).
There are many unwarned people in the mandated evacuation
area with at least 6 reported deaths.
This leaves the question, ``How can one improve the communication
of such high value information?''
The social network is important, and one natural avenue (given that
it is not feasible to communicate to everyone)
is to try to optimize with respect to a fixed budget of people who
you can contact. That is, can we choose the seeds of the information
diffusion to maximize the ultimate spread. This is the question we address.

The study of emergency warnings and evacuation is a good context for
diffusion.
There is no universally
trusted news source, and even if there were, it is practically infeasible
for such a news source to reach everyone.
Thus, it is essential to make use of the social communication network to
diffuse a warning through the network in such a way that people act.
A person hearing the same information
from multiple independent sources is more likely to act on it.
The notion of trust, which
 measures the likelihood that a message from one person to another
will be believed, plays an important role in such diffusions
\cite{Kelton2008,Chopra2003,Chopra2000}.
We will use the model developed in
\cite{Hui2010a} to approximate the process of social information
diffusion. This model captures:
\begin{enumerate}[(i)]
\item
the notion of trust and community structure;
\item
an actor's ability to query friends for more information
(people often receive warning messages from various sources
such as their family and friends, the media or local authorities,
through various channels such as face-to-face or telephone, and may
seek confirmation and/or additional information,
\cite{Sorensen1987,Sorensen1987a});
\item
the existence of multiple sources of the same information, each source 
is trusted
to a different level;
\item
network dynamics as a result of the diffusion (for example nodes evacuating
the network);
\item
tunable thresholds and parameters that can model different
types of diffusion, from gossip in online media to evacuation warnings.
\end{enumerate}
In short, we believe this model to be a reasonable model of general
information diffusion on a social network. We
summarize the essential details of the model in Section~\ref{section:model}.

\paragraph{The Need for Heuristics.}
Given a social network and an information diffusion model, even
one that is much simpler than the model hinted at above, there is no
known
procedure to efficiently\footnote{Choosing a seed set to 
maximize a diffusion
 belongs to the class
of  NP-hard problems, a class 
for which there are no known efficient procedures.
A procedure is efficient if it runs 
in time that is polynomial in the size of the social network. For practical
purposes an algorithm that takes longer than the cube of the network
size is already
 not feasible on population scale networks. Since the only known algorithms
that maximize a diffusion are exponential, such algorithms are far from
feasible.}
 compute a seed set that results in a maximum
number of actors taking the prescribed action. So, our general goal is to 
develop efficient heuristics to compute good (not necessarily optimal)
seed sets on population scale social networks. This means that the heuristics
should have sub-quadratic running time (ideally nearly linear
running time).

\subsection{Our Contribution}
We give efficient heuristics
to select a subset of the actors (the seed set) to initialize with
the information with the goal of trying to maximize
the final set of actors who
believe and \emph{act upon} that
information. We call this a targeting algorithm.
We introduce a new paradigm for
developing such heuristics to construct seed sets in complex and realistic diffusion models.
Our approach works as follows.
First we develop an appropriate simplification of the
general model (with certain tunable knobs) for which
we can develop \emph{efficient, near-optimal targeting algorithms}.
In contrast to the prior work, we develop the
targeting algorithms for a \emph{different, more efficient model},
not the original general model.
The key aspect of our approach
is that the simpler model has tunable parameters;
these tunable parameters can be adjusted so as to optimize the
ultimate diffusion \emph{for the true general model}.
More specifically, there are many instances of the simpler model, each
specified by a particular setting of the tunable parameters. For any
such instance of the simpler model we can compute a near-optimal seed set.
Since we do not know which instance of the simpler model best represents
the true model, we perform a guided search through all instances of the simpler model:
each instance of the simpler model gives a
seed set and we pick the best one.
The details are given in Section \ref{sec:greedy}.
Intuitively, the optimal seed sets for the simpler model give a ``smart'' 
set of seed sets for the general model, and one of these smart seeds is likely
to be good.

\paragraph{Experimental Testing}
We demonstrate the effectiveness of our approach on several random graph models,
as well as a model of the San Diego network during the 2007
fires that was created in \cite{Hui2010b}. In all cases, and for various
types of diffusions (settings of parameters within the true
general diffusion model)
our approach gives the best seeds when compared with two benchmark algorithms:
random targeting and high-degree targeting.

There is a tradeoff, however. Random targeting, though not very effective,
is very easy and requires little information regarding the network
other than the nodes that are present. High degree targeting requires more
information, namely a measure of prominence of the nodes as well, and
often, though not always,  does better than random seeding.
Our algorithms require the most information, namely some model for the
connectivity in the network (who links to whom). Though this may not be
known in practice, it can be estimated as was done in \cite{Hui2010b} for the
San Diego network (based on population densities and a geometric community
based random graph
model that gives highest probability to links that are between
geometrically close nodes within
the same community and lowest probability to links that are between
geometrically distant nodes within different communities).

\subsection{Background and Related Work}

In this section we give a brief summary of related work that forms a pertinent background for our research. Social networks play an important role in the
spread of ideas, information, products, diseases, etc. This diffusion affects the way people think, act and make decisions in a society. Thus modeling the flow of information is an important and active research area.

Infectious diseases are similar to infectious ideas and have dynamics similar to that of evacuation warnings.
Many studies use the homogeneous SIR model of Reed and Frost \cite{Newman2002}. Both Markov chains (e.g., \cite{Chellappa1993}) and differential equations \cite{Bailey1975,Nowak2000} have been used.
Most of these methods make homogeneity assumptions about the underlying network, and thus do not take full advantage of the network topology.

Information diffusion is a long researched area \cite{Valente1996}, with work on online communities becoming a very active topic recently, on account of innovation diffusion, viral marketing, and computer virus spread \cite{Gruhl2004,Hill2006,Iribarren2009,Java2006,Leskovec2007a,Leskovec2006,Strang1998,Wan2007}.
Two fundamental types of models of information diffusion used in the literature are
cascade models and threshold models \cite{Granovetter1978,Young2000,Kempe2003,Kempe2006}.
For example, in the Independent Cascade model each node gets an opportunity to influence each of its inactive neighbors and is successful with a given probability. In the Linear Threshold model, each node has a threshold which determines the number of neighbors that need to be active for this node to become active itself. Such simple diffusion models often have mathematically
convenient properties, like submodularity, which do not hold in
general diffusion models. Submodularity, for instance, implies that efficient
greedy algorithms can construct near-optimal seed sets \cite{Kempe2003}.
Our work, in part, is to use simpler models which are submodular in order to guide the search for good seed sets in the general diffusion model.

\paragraph{Warnings and Evacuation} In order to study the diffusion of warnings and evacuation messages through a population, we need to understand how various psychological, social, and economic factors drive the process.
Humans have a layered warning response process that goes through multiple stages \cite{Sorensen2003}.
Hence, any attempt at realistically modeling this scenario must consider these factors.
For this purpose we turn to work carried out in Social sciences literature.

Much research has focused on employing advanced technology in detection and prediction, or optimizing evacuation routes.
Far less attention has been paid to early warning dissemination \cite{Golden2000,Lindell2007,Lindell2007a,Sorensen2000}.
We conclude from warning response behavior \cite{Mileti1999} and the decision making involved \cite{DisasterResearchintheSocialSciences:FutureChallenges2006,Lindell2004,Perry2003} that it is as much social as it is about communication.
The warning sequence process model \cite{Sorensen2000,Sorensen1987,Sorensen1987a} posits that there are six phases extending from disseminating the warning to protective action.
Social science theory suggests that people's response to warnings depends on the social, economic, demographic and physical/technological environment, \cite{Sorensen2003,Golden2000,Lindell2007a,Aguirre1994,Ding2006} as well as message	 content \cite{Lindell2007a,Sorensen2000,Lindell2004,Perry2003,Sorensen1987,Sorensen1987a}.
People may seek additional information and confirmation through observation and direct contact \cite{Sorensen2000,Mileti1999,Sorensen1987,Sorensen1987a}.
Thus, special care should be taken to incorporate these effects in any general model for diffusion of warnings.
The warning time distribution has two components: the official broadcast component \cite{Perry2003,Rogers1991,Sorensen1992} and the information contagion (diffusion) component \cite{Sorensen1987,Sorensen1987a,Ding2006,Rogers1991}.
There has been considerable research in estimating warning time distributions by modeling the warning network, however the contagion component is little understood \cite{Perry2003,Aguirre1994,Rogers1991,Lindell2004a}, and needs to be better understood for protective agencies to fully estimate warning times \cite{Lindell2007,Lindell2007a}.
Any tool that wants to analyze warning time distribution must incorporate the effects of inhomogeneities such as social groups, ethnic groups, and in general differential trust in and access to warning systems \cite{Sorensen1987,Sorensen1987a,Lindell2004a,Perry1982,Perry2008}. The model of diffusion
we use captures trust as well
as human reactivity and was built around these social considerations.
Our contribution is \emph{not} the model, which we take as given
from prior work
\cite{Hui2010a}. Our contribution is to take this social-science based model
of diffusion and construct good targeting algorithms for it.

\paragraph{Agent based models}
In scenarios where individuals are influenced by their social environment, agent based modeling is often used \cite{Delre2007,Gruhl2004}.
In such cases, decision making entities are represented by agents whose behavior can be based on a set of rules.
Such a model allows the agent to have intelligence and memory and also exhibit complex behaviors like learning and adapting \cite{Bonabeau2002,Monge2003}.
In the context of diffusion, agent based modeling has been used, for example,  to study epidemic spread \cite{Bass1969,Perez2009}; to model information diffusion in virtual marketplace \cite{Neri2004}; and to simulate technological diffusion \cite{Amman2006,Ma2005} and environmental innovations \cite{Berger2001,Schwarz2009}.
Our research is not in the development of such models but in the
exploitation of such models in determining optimal seed sets.

\subsection{Outline of the Paper}
In the next section, we summarize the diffusion model from 
\cite{Hui2010a} and
give some analysis of its complexity in Section~\ref{sec:analysis}.
We present our projected greedy heuristic in
Section~\ref{sec:greedy}. In Section~\ref{sec:design}
we describe the experimental
design we used in extensively evaluating our projected greedy
heuristic on both synthetic networks as well as the San Diego network. We end
in Section~\ref{sec:disc}
where we give an overview of the results together with
a discussion.

\section{The Diffusion Model}
\label{section:model}

We summarize the diffusion model that was developed
in \cite{Hui2010a}. This model was used
to study propagation of evacuation news through a population
in~\cite{Hui2010a}.
We use it as our model of diffusion because it is general and contains many
other simpler models as sub-cases.
The model takes into consideration the fact that agents may act on information and leave the network; this means diffusion occurs on a
network that is dynamic. The model contains the
 concepts found in the widely studied SIR model of epidemiology, as well as standard threshold and cascade models together with actions a social agent may
perform, such as ``querying for more information'' as well as notions of 
trust.

We have an initial graph $G=(V,E)$ where $V$ are individuals in the
population and the edges $(u_i,u_j) \in E$ represent social connections
between individuals $u_i$ and $u_j$.
There are $K$ sources with \emph{information values}
$I_{1},I_{2},\hdots,I_{K}$.
Also, a source $k \in \{1,\hdots,K\}$ can seed $B_{k} \geq 0$ nodes.
We will use \math{k} to index sources and \math{i,j} to index nodes.

Associated with edge $(u_i,u_j)\in E$ is a trust value $0 \leq \alpha (u_i,u_j)\leq 1$. It represents the amount of trust node $u_j$ has on information provided by $u_i$. The graph may be directed with different trust values on edges $(u_i,u_j)$ and $(u_j,u_i)$ representing asymmetrical trust.
Each node $u \in V$ also has a similar trust value for each
source $k \in \{1,2,\hdots,K\}$ : $\alpha(u,k)$.

Each node $u$ possesses an {\em information-value set} $S(u) = \{(s_{1},v_{1}),\hdots,(s_{K},v_{K})\}$ for the $K$ information sources. It is made up of pairs
 $(s_{k},v_{k})$, where $v_{k}$ is the information node
\math{u} has from source $s_{k}$. Based on its information value set, a node calculates its {\em information value} as follows:
\mld{I(u) = \lambda_{d} \sum_{k = 1}^{K}v_{k} +
(1-\lambda_{d})\cdot \max_{k = 1,..,K}v_{k},\label{eq:Iu}
}
where $0\leq \lambda_{d} \leq 1$ is a model parameter.
The information value is a convex combination of the total information the node
has from all the sources and the maximum value the node has from any one source;
as such, \math{I(u)} is
at most this total information, and at least
the maximum (a node's information is at least the information of its
best source and at most the sum of all its sources).
The two extremes \math{\lambda_d=0} and \math{\lambda_d=1}
correspond to two different extremes of a diffusion. When
\math{\lambda_d=0}, a node ignores all information but its highest information
value source; this corresponds to a very conservative choice and
would be more appropriate for something like a warning.
When
\math{\lambda_d=1}, the information value is the sum, i.e., the
different sources reinforce each other. This is a more aggressive
diffusion, appropriate for something like gossip or a rumor.

Associated with each node $u$ are two thresholds, a lower threshold
$t_{l}(u)$ and an upper threshold
$t_{h}(u)$ such that $0\leq t_{l}(u) \leq t_{h}(u)$.
The thresholds determine how
a node acts given its information set.
Based on its information value,
a node $u$ lies in one of the following states:

\begin{itemize}
\item[1.] {\bf Disbelieved (low information value, less than the
lower threshold)}: If $I(u) < t_{l}(u)$ then the node does not believe it has any meaningful information. In this state, node \math{u} does
not take any action except for incorporating information-value sets
that it (node \math{u}) receives
from its neighbors.

\item[2.] {\bf Undecided (intermediate information value)}: If $t_{l}(u) \leq I(u) < t_{h}(u)$ then node \math{u}
believes it has some information but is uncertain about acting (which in our case is to
evacuate).
In this state, node \math{u} will
 \emph{query} the information-value sets of its neighbors and
incorporate any new information into its own information set.

\item[3.] {\bf Believed (high information value, above the upper
threshold)}: If $I(u) \geq t_{h}(u)$ then the node has enough information to accept the information as correct and act upon it. In this state, node \math{u}
actively propagates its information-value set to its neighbors for $\tau$ time-steps before evacuating. Here $\tau>0$ is a model parameter that determines the time it takes a believed node to evacuate the network (in the case of
a warning message). If \math{\tau=\infty}, the node never leaves the network
and propagates its information forever.

\item[4.] {\bf Evacuated}: $\tau$ time-steps after a node enters the Believed state, it evacuates resulting in its disconnection from the network. Hence the graph changes and a new graph $G'$ is obtained. The new graph \math{G'} is the induced subgraph on the vertices \math{V\setminus u}; that is
 $G' = (V',E')$, where $V' = V\setminus u$ 
and $E' = E\setminus\{(u,u_i)\,|\, u_i \in V', \, (u,u_i) \in E\}$.
\end{itemize}

\subsection{The Diffusion Process}

Initially, all nodes have zero information, so
$v_{k} = 0 \ \forall k \in \{1,\hdots,K\}$
  and all nodes are in the 
disbelieved state. Each node's information set is zero,
\mand{
S(u)=\{(s_1,0),\ldots,(s_K,0)\}.
}
The diffusion process begins when some nodes are \emph{seeded}
 with information from sources according to a seeding strategy.
Seeding entails transfer of information from sources to the information-value
set of the selected seed nodes.
This transfer of information gets attenuated by the trust between
the node and the source.
So if source $k$ seeds node $u$ then the information transfer to node $u$ has
 value $\a(u,k)\cdot I_{k}$ and the pair
\math{(s_k,0)} in the node's information-value set gets updated to
\math{(s_k,\a(u,k)\cdot I_{k})}.
Each source seeds some subset of the nodes in this way, and multiple
sources can seed the same node.
After seeding,
each node then computes its information value using \r{eq:Iu}, ascertains its state, and performs the required action at the next step.
At every consecutive step nodes take action
depending on their state (which may involve broadcasting its
information-value set and/or receiving information-value sets from its
neighbors); if any new information is received, this information is merged
into its current information-value set. In this way a node's information-value
 set
evolves as the diffusion process proceeds.

We now describe the process of propagating and updating information sets.
When node $u_i$ propagates its information-value set to neighbor $u_j$,
the value of information gets similarly attenuated by a
factor of $\alpha(u_i,u_j)$. If
\mand{
S(u_i)=\{(s_1,v_1),\ldots,(s_K,v_K)\},
}
then the information-value set that is received by \math{u_j} is
\mand{\a(u_i,u_j)\cdot S(u_i)=\{(s_1,\a(u_i,u_j)\cdot v_1),\ldots,(s_K,\a(u_i,u_j)\cdot v_K)\}.
}
Consider a source \math{k}, and consider node \math{u_i}.
The node \math{u_i} may receive information value originating from
source \math{k} either directly from source \math{k} or indirectly from one
of its neighbors. Suppose that node \math{u_i} has
\math{\delta} neighbors. In principle \math{u_i} can receive (either in the
current step or in some prior step) information-value
\math{v^0_k} (directly from the source) and information values
\math{v^1_k,v^2_k,\ldots,v^\delta_k} from each of its neighbors (each of
these information
values will be the attenuated value that was propagated to \math{u_i};
some or all of these values could be zero).
In order to determine its information value for this source \math{k}, the node
needs to \emph{fuse} all these information values into a single value as
follows
\remove{
 In this way after propagation a node might end up with multiple
source-value pairs containing
information from \emph{the same source}, including the information it
already contained. For example if all neighbors of \math{w} have information
value for source \math{s_1}, they may each propagate this information to
\math{w}.
 Let's say that node $w$ has received $m$ source-value pairs (including
its own one) containing source $s_{k}$ with associated values: $v_{k}^{1},v_{k}^{2},\hdots,v_{k}^{m}$. Then node $w$ will construct a \emph{fused}
 source-value pair $(s_{i},v_{i}^{*})$ as follows:
}
\mld{
v_{k} = \lambda_{s}\cdot \sum_{j =0}^{\delta}v_{k}^{j} +
(1-\lambda_{s})\cdot \max_{j=0,..,\delta}v_{k}^{j}.
\label{eq:fused}
}
Again $\lambda_{s}$ is a model parameter such that $0 \leq \lambda_{s}\leq 1$.
This fusion of information happens for every source \math{k} at every node
\math{u_i}.
As with \math{\lambda_d}, \math{\lambda_s} impacts how aggressive a 
diffusion is. At the two extremes:
\math{\lambda_s=0}, a node takes the maximum value it hears about the
information from a source (the conservative case); and,
 when \math{\lambda_s=1} a node takes all the information it hears about a 
source and adds (the aggressive diffusion).
Thus, at each consecutive time step, every node updates its information-value
set based on the new information that was propagated to it. A node calculates
its new information value based on this updated set and this updated information
value will be used to determine the node's state and possible action.
The diffusion process continues, i.e. information continues
to propagate as described above, until either all nodes evacuate or
there is no change in the information value sets of the nodes.

The graph \math{G} is typically given and chosen to model some social network,
as for example the San Diego network during the 2007 fires.
The model parameters, $(\lambda_{s}, \lambda_{d})$, the evacuation time
\math{\tau},
the threshold parameters \math{t_l(u),\ t_h(u)} at each node 
\math{u} and the
edge trust values \math{\a(u_i,u_j)} can be chosen to model various
types of diffusion settings. For example in a social network with
strong communities, the trust values of edges within a community would be
high (close to 1) and the edges between communities would typically be
low (close to 0).
In a network that has been ``primed'', for example a community that has
recently experienced a tsunami, the community may be in a ``panic'' state,
which could be modeled by very low upper thresholds \math{t_h(u)}. One would
also set the thresholds to be low for low risk diffusions like gossip.
However for  actions that incur significant cost, like a node evacuating, the
lower threshold may be low but the upper threshold would typically be high.
Lastly the \math{\lambda}-parameters could be chosen to model fast or slow
diffusing information.
In all cases, however, all these parameters are exogenously specified.
An instance of the general diffusion model \math{\mathbb{G}} is
summarized in Figure~\ref{fig:instance}.
Our goal is to optimally \emph{seed} the diffusion given
the graph and the parameter settings.
\begin{figure}
\begin{center}
\fbox{
\parbox{4.55in}{
{\bf Instance \math{\mathbb{G}}}
\\
Graph \math{G=(V,E)} specifying the network for the diffusion.
\\
Source information values and budgets, 
\math{(I_1,B_1),\ldots,(I_K,B_K)}
\\
Trust values \math{\a(u_i,u_j)} for all edges \math{(u_i,u_j)\in E}.
\\
Trust values \math{\a(k,u_j)} between each source \math{k} and each node
\math{u_j\in V}.
\\
Diffusion parameters \math{\lambda_d,\lambda_s,\tau}
\\
Lower and upper thresholds
\math{t_l(u),t_h(u)} for each node \math{u\in V}.
\\
{\bf Desired output:} seed sets \math{\p_1,\ldots,\p_K} for each source 
\math{k=1,\ldots,K}.
}}
\end{center}
\caption{An instance of the general diffusion model.\label{fig:instance}}
\end{figure}

\subsection{Seeding the Diffusion}
The goal of this study is to find a seeding strategy that maximizes the number of nodes that end up in their Believed state, given the diffusion setting
as described in the previous section, and a seeding budget, as we describe here.

Each source \math{k} has a budget \math{B_k} of nodes which it can seed. Our
task is to determine the seed sets \math{\p_k\subset V}, where
\math{|\p_k|\le B_k}. The seed set \math{\p_k} specifies which nodes
are seeded by source \math{k}. The objective is to maximize the
number of nodes that become believers. Thus, there is some function
\math{\g}, which we call the \emph{coverage function}, that computes
the number of nodes which become believers, given the diffusion setting
and the seed sets \math{\p_1,\ldots,\p_K}. More generally, if there
is some randomization in the communication (propagation of information-value
sets)\footnote{For example if, when a node \math{u_i}
 propagates an information-value set to \math{u_j}, it is
received at the other end with some probability \math{p(i,j)} depending on the
communication infrastructure, then the number of nodes which are ultimate
believers is a random variable.}, then
the coverage function would compute the expected number of believers. Thus,
in general, the coverage function maps
\math{(\p_1,\ldots,\p_K)} to \math{\mathbb{R}_{\geq 0}}.

One of the challenges is to efficiently compute the coverage function
\math{\g} for a given input seeding \math{(\p_1,\ldots,\p_K)}.
The other, which is our main goal,
is to find the seeding which maximizes
\math{\g}.

\remove{
For this purpose we define a {\em coverage function} on the nodes $\g : ({2^{V}})^{K}\rightarrow \mathbb{R}_{\geq 0}$. The input of this function is $K$ sets of nodes, each of size at most $S_{i}$, that are seeded by the source nodes. It gives the number of nodes converted to Believed state as output. Therefore, the objective is to find $K$ sets of nodes $\{\p_{1},\p_{2},\hdots, \p_{K}\}$ where each $\p_{i} \in V$ and $|\p_{i}| \leq S_{i}$, such that $\g(\p_{1},\p_{2},\hdots, \p_{K})$ is maximized.
}

\section{Analysis}
\label{sec:analysis}

It is useful to have a theoretical understanding of the diffusion
process in order to identify where the potential difficulties
lie when choosing an optimal seed set.
In fact, as we will soon see, the general model  described
in the previous section is extremely difficult for theoretical analysis.
As a result,
we will consider a simplified instance of the diffusion
model in our theoretical analysis, and use the insight from this analysis to develop a seeding algorithm for the general model.
When testing the algorithm, however, we will use instances of the general
diffusion process (see Section \ref{section:model}).

\subsection{Monotonicity and Submodularity}

The complexity of seed selection for the general diffusion model
results from the behavior of the coverage function \math{\g} when you perturb
the seed set. If \math{\g} behaves well when you perturb the seed set
in certain ways, then simple, iterative greedy algorithms are effective
at selecting near-optimal seed sets.

\paragraph{Monotonicity.}
The coverage function \math{\g} is monotone if adding nodes to a
seed set can only increase the coverage function's value. Mathematically,
if \math{\p\subseteq\p'}, then
\mld{
\g(\p,\p_{-i})\le\g(\p',\p_{-i}).\label{eq:mon}
}
(We use the notation \math{\g(\cdot,\psi_{-i})} to denote the coverage function
as a function of its \math{i}th set argument, keeping all the other sets
fixed.)
It is intuitive that if you seed more nodes in an evacuation scenario, then
more people should evacuate.

\paragraph{Submodularity.} Submodularity captures the intuitive notion of
diminishing returns. As you seed more and more, the additional benefit you
get in terms of the increased coverage is decreasing.
One of the equivalent mathematical
definitions of submodularity is as follows.
\math{\g} is submodular, if for all sets $\p\subseteq\p'$ and
any \math{A\subseteq V},
\mld{
\g(\p\cup A,\p_{-i})-\g(\p,\p_{-i})
\ge
\g(\p'\cup A,\p_{-i})-\g(\p',\p_{-i}),\label{eq:submod}
}
that is the increase in \math{\g} from \math{\p'\rightarrow\p'\cup A} is not
more than its increase from \math{\p\rightarrow\p\cup A}.
Submodularity is essentially the set-function version of concavity.

A coverage function that is submodular has the nice property that if you
use a simple greedy strategy to select a seed set, then the resulting
seed set has a coverage which is within a small constant factor 
of optimal. Unfortunately, as we are about to demonstrate, the coverage function
for our general diffusion model is not submodular. It is not even
monotone. The next two examples illustrate why. In both examples we
set \math{\lambda_d=\lambda_s=0}, so we are using the max function for
both the computation of the information value and for the fusion
of information received from neighboring nodes on the same source; we sometimes
call this the \emph{max-max} model. We have only one source (\math{K=1}) and its
information value is \math{I_1=1}.

\paragraph{Example 1: non-monotonicity:}  The diffusion setting is shown
in Figure~\ref{fig:monotonicity}. There is one low trust
edge (indicated in red) and a single
``high strung'' node $c$ which has very low thresholds (also indicated in
red); we will soon see why we call this node
high strung.
\providecommand{\st}[1]{\scalebox{3.5}{#1}}
\providecommand{\stn}[1]{\scalebox{2.25}{#1}}
\providecommand{\stb}[1]{\scalebox{5}{#1}}
\providecommand{\sth}{\scalebox{2.5}{\math{\left.\begin{array}{c}\\\\\\\\\\\\\end{array}\right\}}}\ \ }
\begin{figure}[h!]
\begin{center}
\begin{tabular}{m{2.5in}@{\hspace*{0.5in}}m{2.5in}}
\resizebox{2.5in}{!}{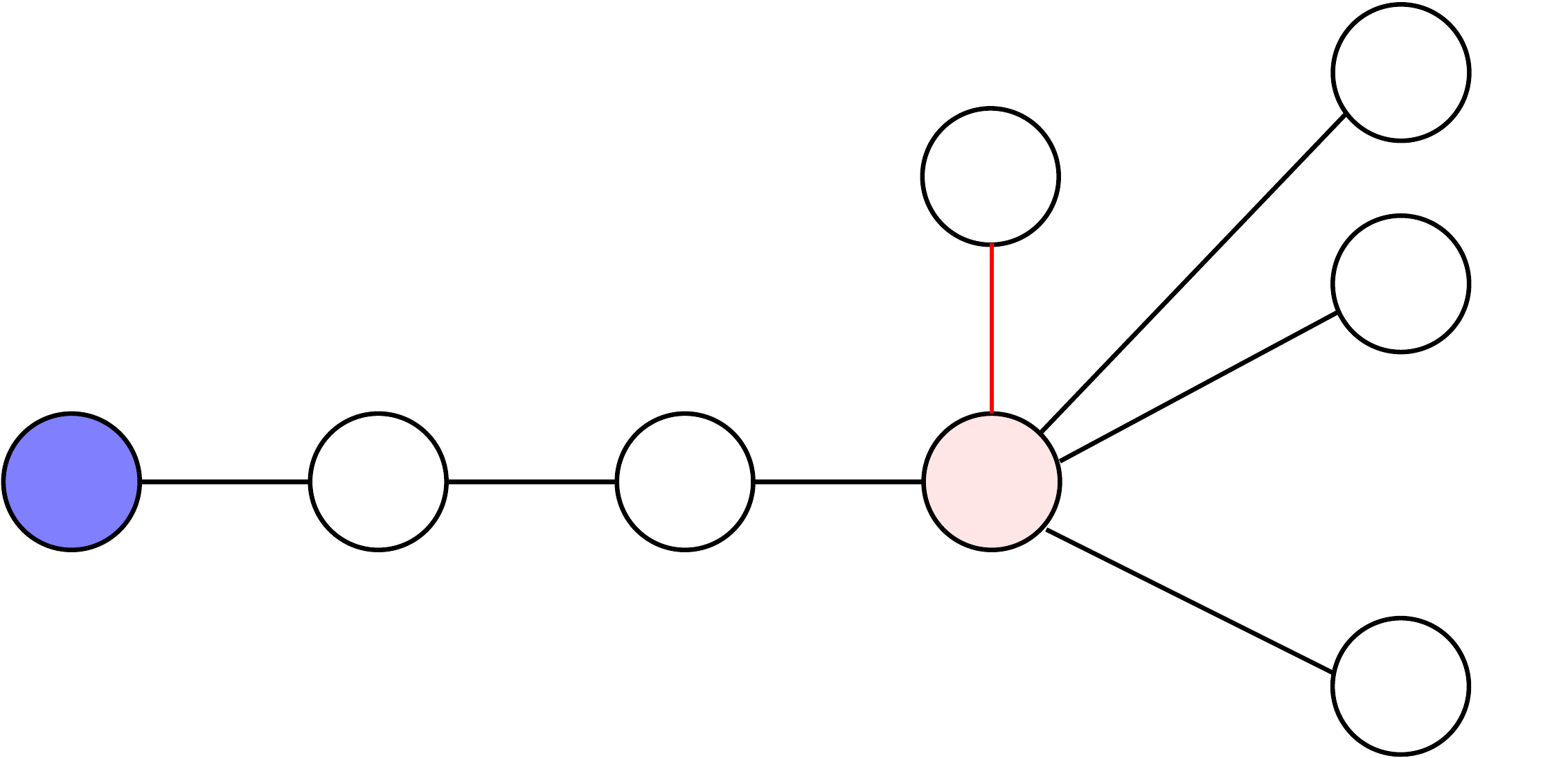}
&\resizebox{2.5in}{!}{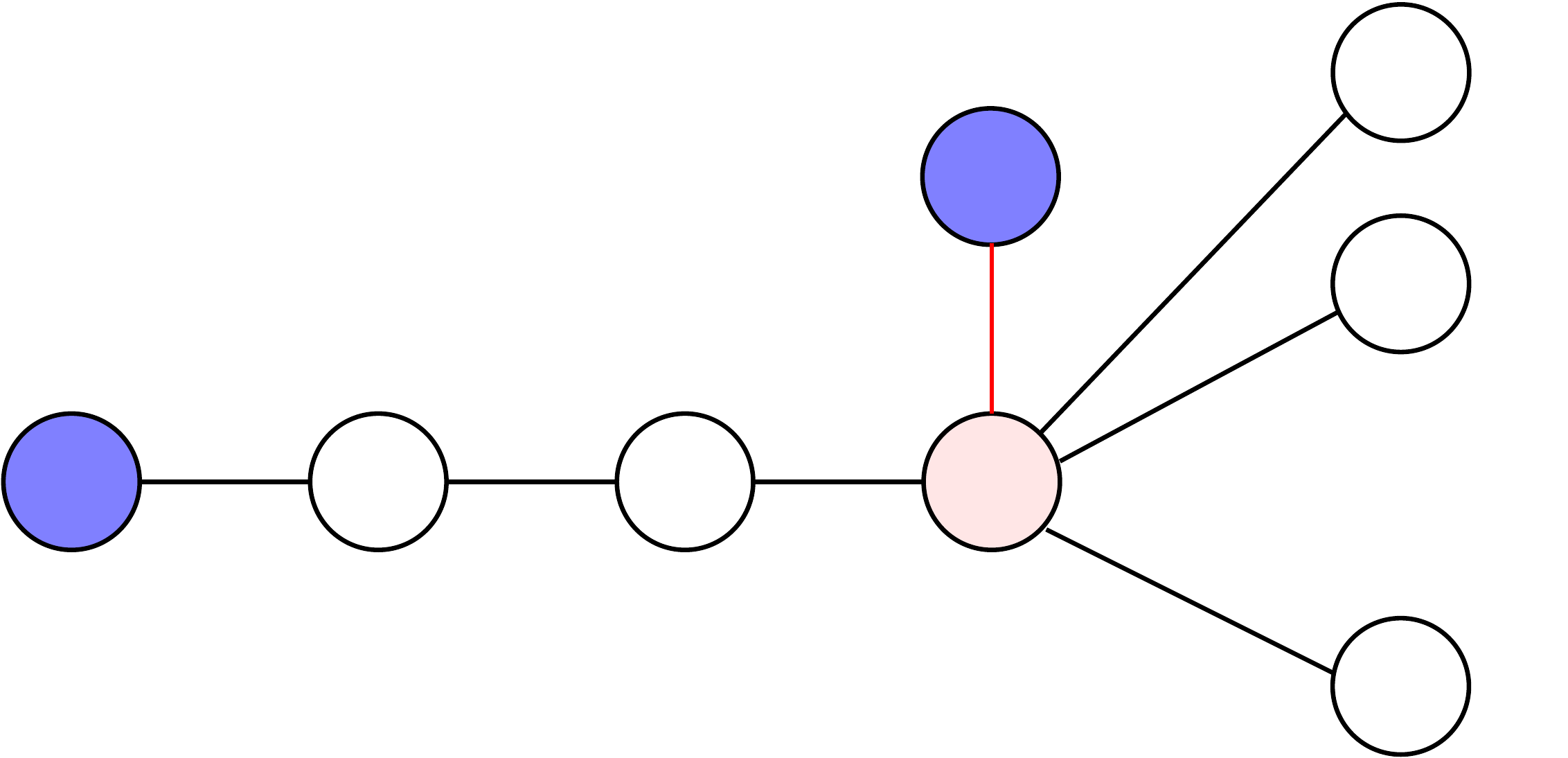}\\
\centerline{\scalebox{2}{\math{\downarrow}}}&
\centerline{\scalebox{2}{\math{\downarrow}}}\\
\resizebox{2.5in}{!}{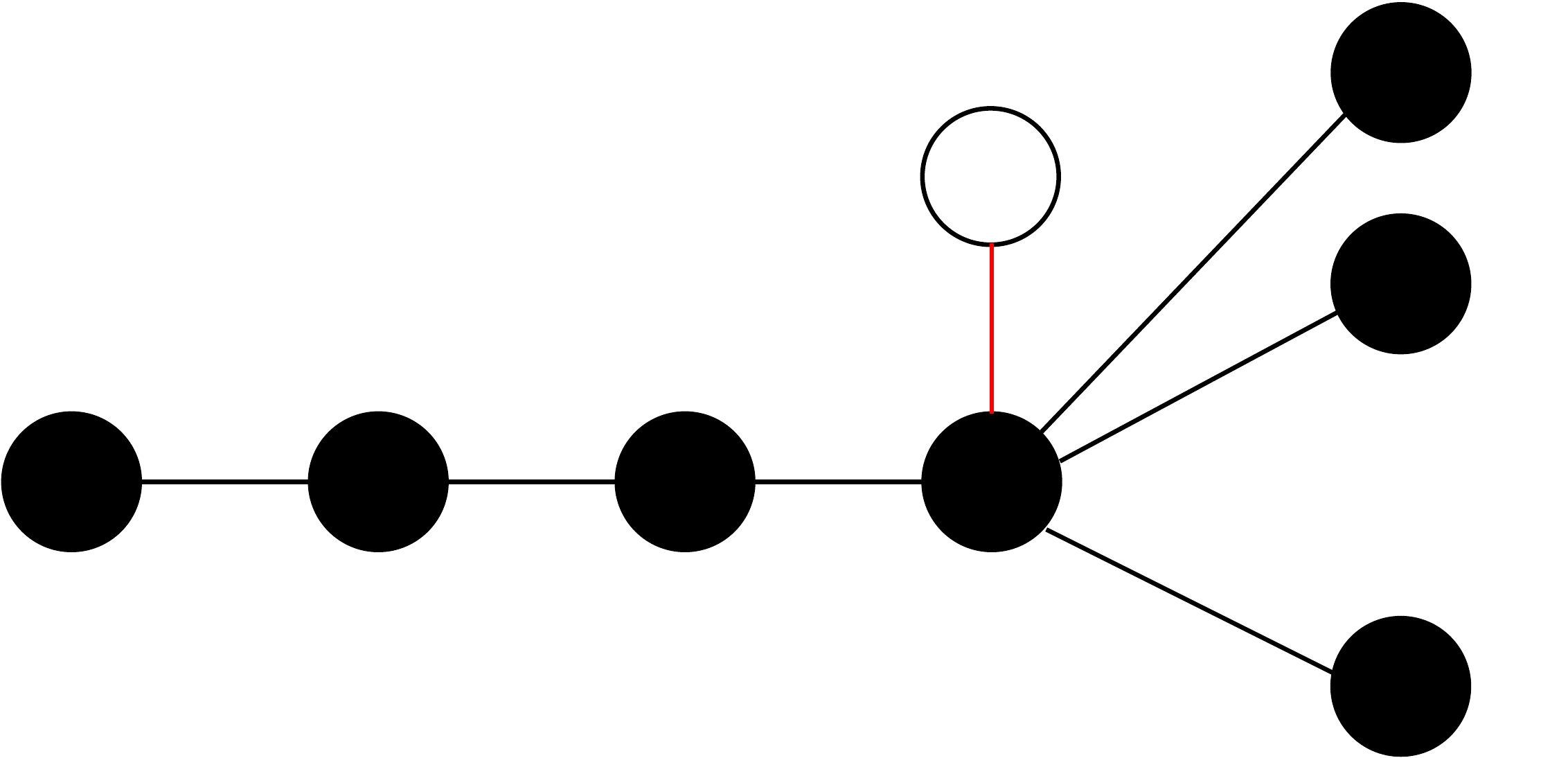}
&\resizebox{2.5in}{!}{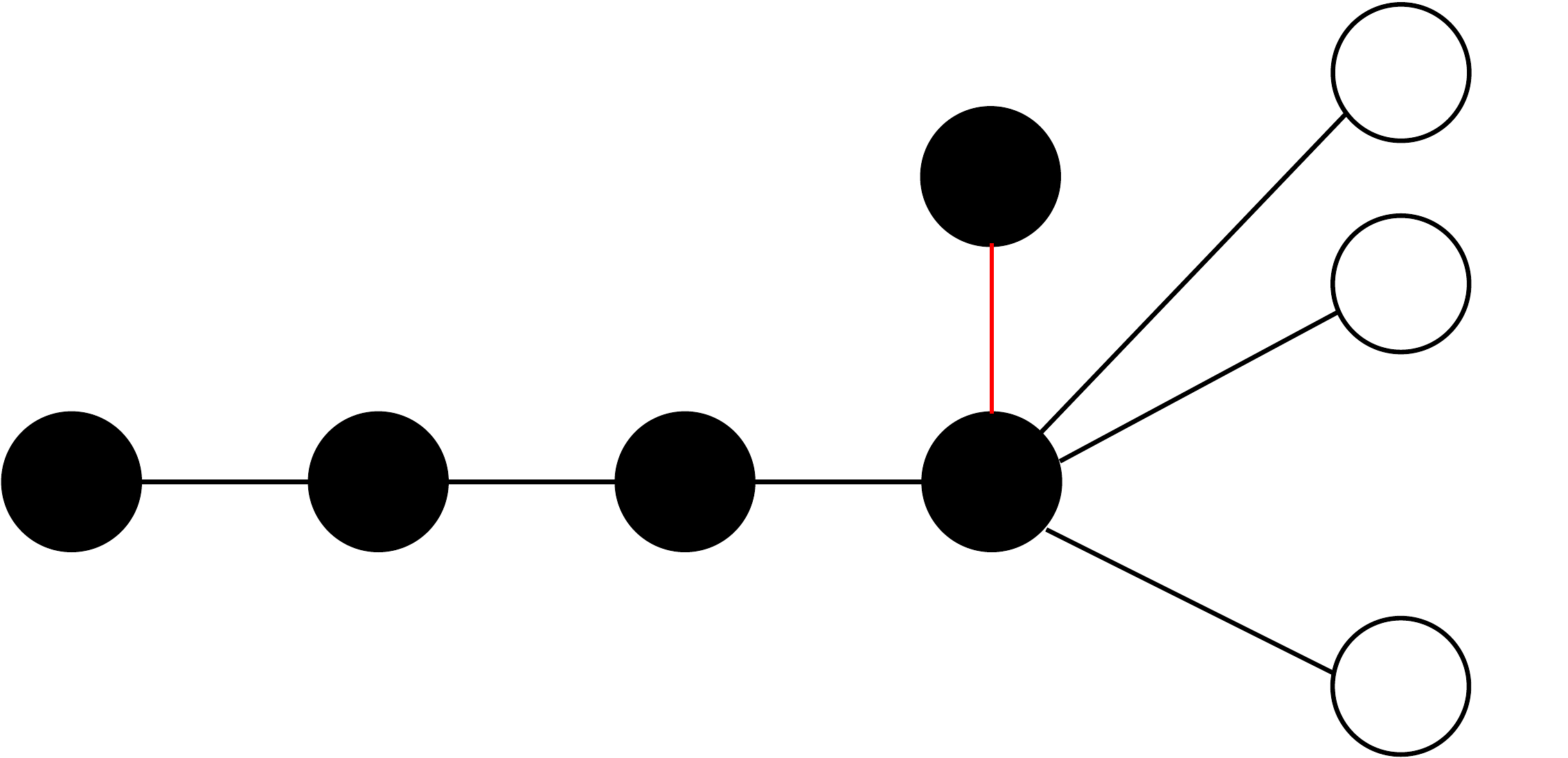}\\
\centerline{(a)}&\centerline{(b)}
\end{tabular}
\caption{\small
Non-monotonic diffusion setting. The blue nodes are the initial
seed set. The black nodes are the evacuated nodes in the final state. All
thresholds are 0.5 except for the high strung node \math{c} with threshold
\math{0.1}. The number of steps to evacuation is \math{\tau=1}.}
\label{fig:monotonicity}
\end{center}
\end{figure}

Consider the coverage function for this diffusion setting and
two different seed sets
$\p = \{a\}$ (case (a) in Figure~\ref{fig:monotonicity})
and $\p' = \{a,b\}$ (case (b) in the Figure~\ref{fig:monotonicity}) --
in Figure~\ref{fig:monotonicity}) the seed sets are
shaded blue for both cases. In case (a), the information value of
\math{1} will propagate unattenuated to node \math{c} and beyond to the
additional \math{k} nodes. All these nodes will evacuate. Only an
information value of \math{0.1} propagates to node \math{b} which will
not evacuate.
So, $\Gamma(\p) = k+4$.
Now consider case (b) with seed set $\p'$.
Initially node \math{c} will receive attenuated information of value
\math{0.1} from node \math{b}. Since node \math{c} has very low thresholds,
in a sense it panics, and leaves the network almost immediately
(\math{\tau=1}),
after propagating its information value of 0.1. Unfortunately, however,
this value of 0.1 is not high enough to evacuate
 any of the \math{k} peripheral nodes, since their thresholds are high.
The nodes between \math{a} and \math{c} are fine, though, because they
will eventually receive an information value of \math{1} propagated from
\math{a}. The important thing is that when \math{c} leaves, it cuts off the
\math{k} peripheral nodes from \math{a} which results in only
\math{5} nodes evacuating.
Hence, $\Gamma(\p') = 5$.

The high strung node \math{c} will leave if it gets only a small amount of information.
If, as in this case, the high strung node is crucial (is a ``bridge'' node in the network), then when the high strung node leaves, it disconnects potentially
large parts of the network from the information before the higher
value information has a chance to flow into those parts of the network.
By increasing the seed set, one might increase the chances that such a
high strung node gets low-value information too early in the diffusion,
and the ensuing early evacuation of this high strung node is what leads to the
non-monotonicity of the coverage function.
In fact, if nodes do not leave the network until all information propagation has
occurred, then monotonicity is guaranteed. As the next example will show, however, even if nodes do not leave the network,
the coverage function is still not submodular.

\paragraph{Example 2: non-submodularity.}
The diffusion setting is
shown in Figure~\ref{fig:submodularity}. We set the evacuation time \math{\tau=10} (i.e., large
enough so that nodes only start leaving after all information has propagated).
\providecommand{\st}[1]{\scalebox{3.5}{#1}}
\providecommand{\stn}[1]{\scalebox{2.25}{#1}}
\providecommand{\stb}[1]{\scalebox{5}{#1}}
\providecommand{\sth}{\scalebox{2.5}{\math{\left.\begin{array}{c}\\\\\\\\\\\\\end{array}\right\}}}\ \ }
\begin{figure}[h!]
\begin{center}
\begin{tabular}{m{2.5in}@{\hspace*{0.5in}}m{2.5in}}
\resizebox{2.5in}{!}{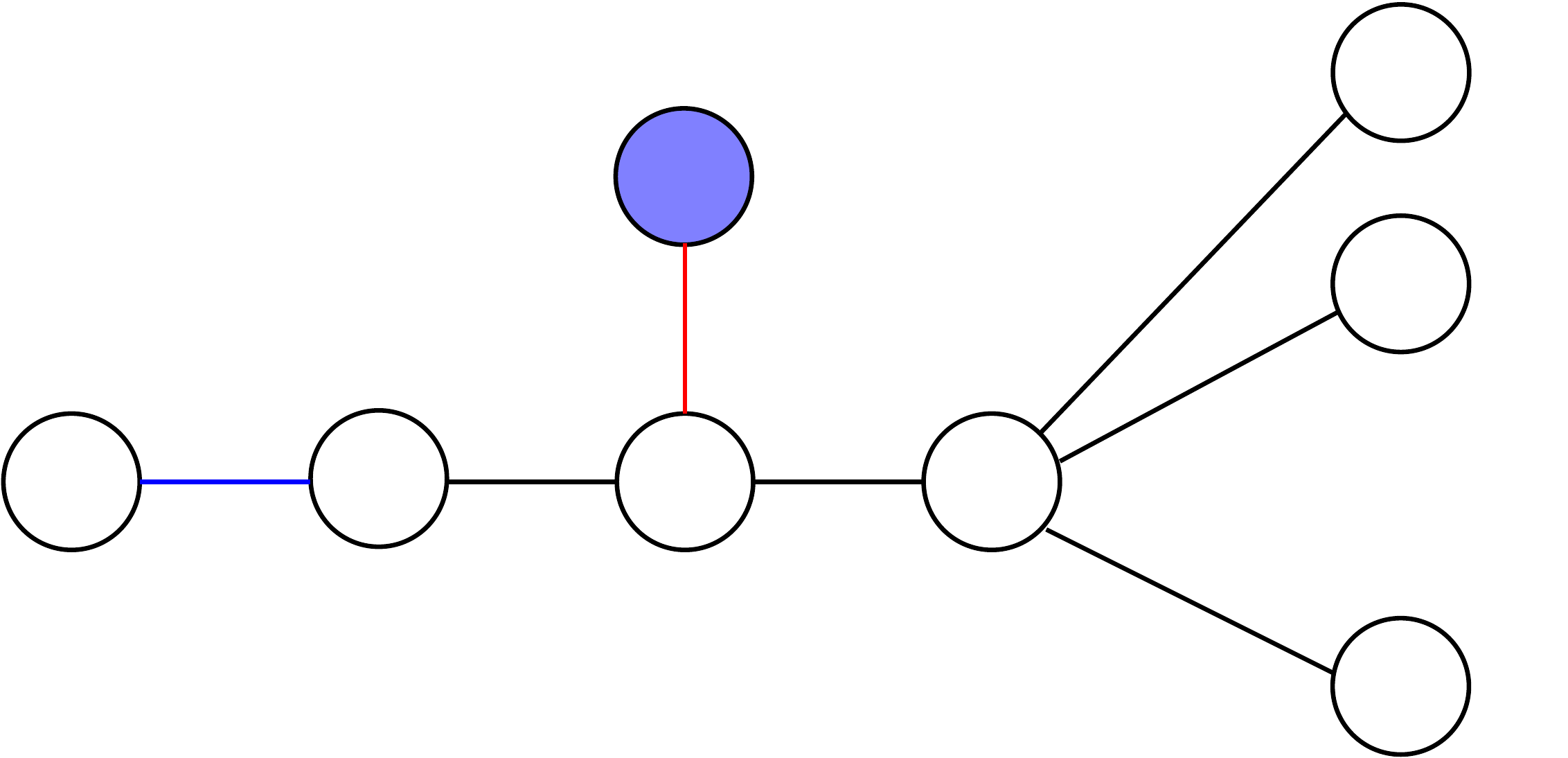}
&\resizebox{2.5in}{!}{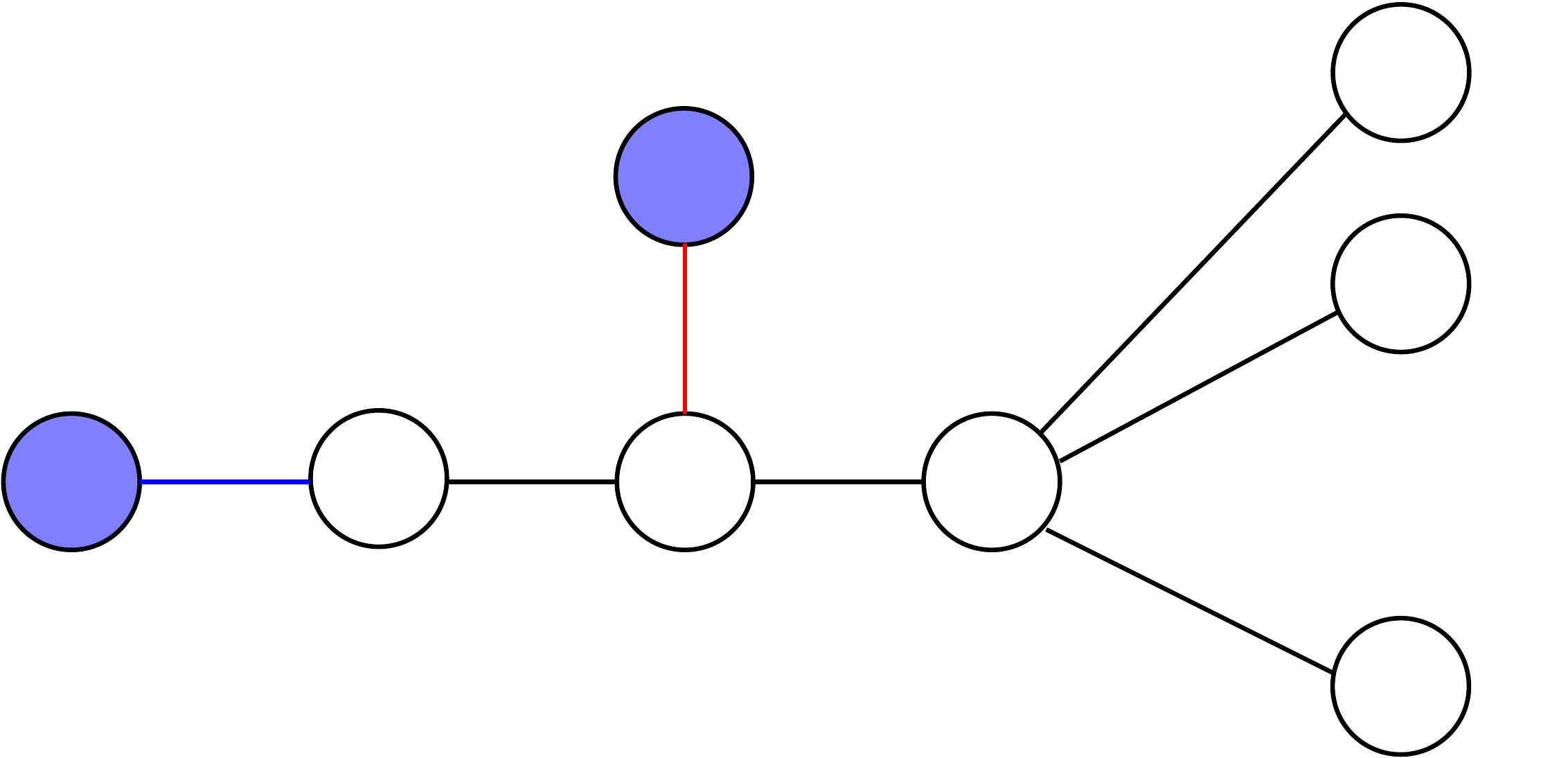}\\
\centerline{\scalebox{2}{\math{\downarrow}}}&
\centerline{\scalebox{2}{\math{\downarrow}}}\\
\resizebox{2.5in}{!}{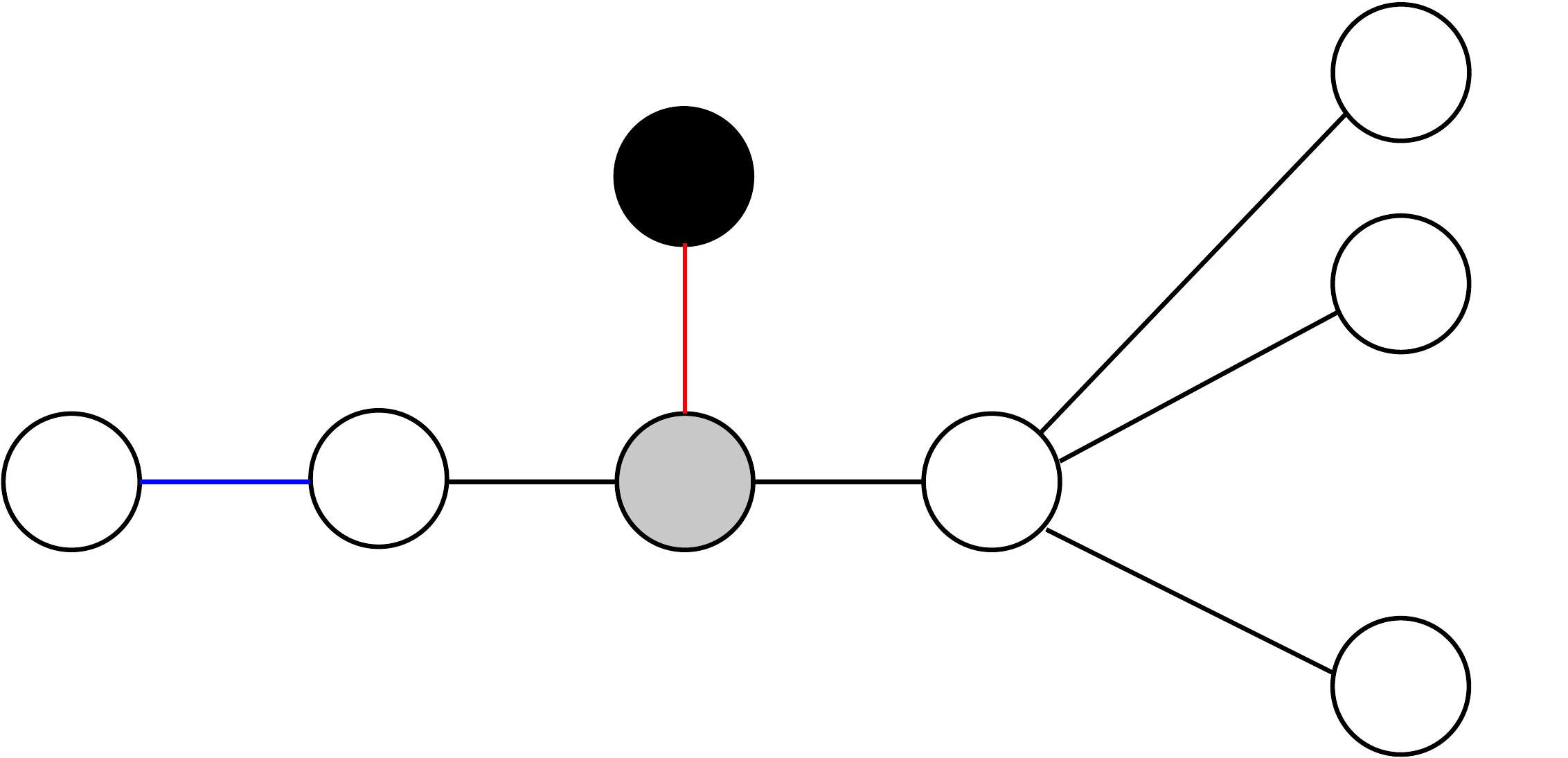}
&\resizebox{2.5in}{!}{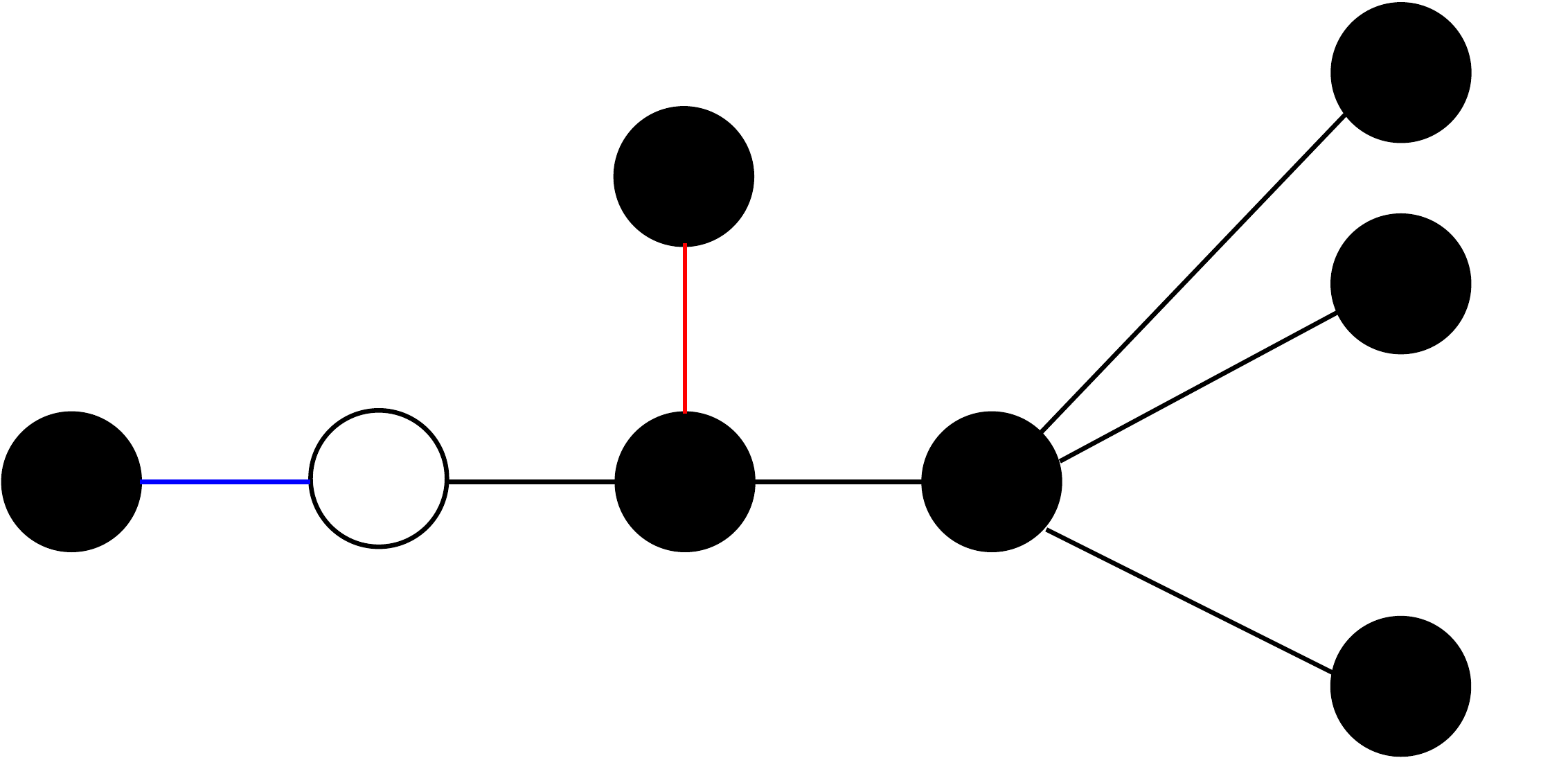}\\
\centerline{(a)}&\centerline{(b)}
\end{tabular}
\caption{\small
Non-submodular diffusion setting. The blue nodes are the initial
seed set. In the final state, the black nodes are the evacuated and the light gray are undecided . All
thresholds are 0.5 except for nodes \math{c} and \math{d}: for node \math{d},
$t_{l}(d)=t_{h}(d)=0.91$ and for node \math{c},
$t_{l}(c)=0.1$, $t_{h}(c)=0.9$.
The number of steps to evacuation is \math{\tau=10}.}
\label{fig:submodularity}
\end{center}
\end{figure}

Now consider the seed sets $\p=\emptyset$, $\p' = \{a\}$.
Clearly \math{\g(\p)=0}. As for \math{\p'}, if you seed \math{a},
an attenuated information of \math{0.9} reaches \math{d} which is not enough
to breach the threshold of \math{0.91}. Hence only \math{a} evacuates and
so \math{\g(\p')=1}. We now consider
\math{\p\cup\{b\}} and \math{\p'\cup\{b\}}. We will see that
\mld{\label{eqn:inequality}
\g(\p\cup\{b\})-\g(\p) < \g(\p'\cup\{b\})-\g(\p'),
}
which contradicts submodularity, since \math{\p\subseteq\p'}.
First consider \math{\p\cup\{b\}=\{b\}} which is case (a)
in Figure~\ref{fig:submodularity}.
An information value of \math{0.1} propagates to \math{c} which will therefore
enter the undecided state since \math{t_l(c)} is low. However, the neighbors
cannot provide any new information so nothing further will happen. Therefore,
only \math{b}
evacuates.

Now consider \math{\p'\cup\{b\}=\{a,b\}} which is case (b)
in Figure~\ref{fig:submodularity}.
As before, \math{0.9} propagates to \math{c} and \math{0.9} to \math{d}. But,
because \math{c} is in the undecided state, it will \emph{query} 
\math{d} and
receive information of \math{0.9} which is enough to push \math{c}
to become a believer. Now, \math{c} will propagate its information value of
\math{0.9} to \math{e} and beyond.
Node \math{d}, however, will always remain below its threshold of
\math{0.91}. Hence \math{\g(\p'\cup\{b\})=k+4}.
So
\mand{
\g(\p\cup\{b\})-\g(\p)=1<k+3=\g(\p'\cup\{b\})-\g(\p').
}
The fact that node $c$ could query for information that $d$
did not transmit is what resulted in the cascade effect
which converted node \math{e} and beyond to the
 Believed state. Though it is natural for humans in a social diffusion
process to query, it is precisely this ability to query that leads to
non-submodularity and the resulting complexity of the process.

\subsection{A Simplified Model}
\label{section:simplified}

Taking a cue from the examples above, we define a simplified model that is a
strict subcase of the general diffusion model defined earlier. Note that while doing experimental simulations we will use the general diffusion model. We use the simplified model only to carry out theoretical analysis.
Insights obtained for this simplified model help us develop heuristics that we then apply to the general case.

In this simplified
 model $\lambda_{d}=\lambda_{s} = 0$. So, the information value at a node
is the maximum value among source-value pairs, i.e.,
$I(u) = \max_{k=1,..,K}v_{k};$
similarly, when fusing information arriving via different neighbors regarding
the same source we take the maximum; we call this the
\emph{max-max} model.
We also assume all
 nodes have a single threshold, so  $t_{l}(u) = t_{h}(u)=t(u)$.
This eliminates the Undecided state which, as we saw,
leads to non-submodularity in the coverage function.
To ensure monotonicity, we choose the evacuation time \math{\tau} to be
sufficiently large (e.g., $\tau\gg|E|$).
This means that nodes do not leave the network until all information
propagation has occurred. For this simplified \emph{max-max} model, we prove
a simple lemma that characterizes the nature of the coverage function. We
are interested in
\mand{
\g(\p_1,\ldots,\p_K).
}
Define the \emph{singleton} coverage (set) function
\math{\gamma(u,k)} to be the \emph{set} of nodes ultimately
 converted to the Believer state when
only one node \math{u} is seeded by a source $k$
with information value \math{I_k}.
Thus, $\g(\p_1,\ldots,\p_K)$ equals $|\gamma(u,k)|$ when $\p_i=\emptyset$ for all $i\neq k$, and $\p_k=\{u\}$.
The next lemma characterizes the form
of \math{\g(\p_1,\ldots,\p_K)}. It basically says that to obtain
the set of nodes
converted to the Believer state, it suffices to consider
each source and let it just seed one of the nodes in its seed set,
with no other source seeding any nodes.
Some set of nodes is converted to Believer in this
process. We consider, in this way, all the (source,seed) pairs in turn,
computing the converted set when just this source seeds just this one
node. By taking the union of all these converted sets, we get the
final set of nodes converted when all the sources simultaneously seed
all the nodes in their respective seed sets.
\begin{lemma}\label{lem:Gamma}
The set of nodes converted to Believer with seed sets \math{\p_1,\ldots,\p_K}
is
\mand{
\bigcup_{k=1}^K\bigcup_{x\in\p_k}\gamma(x,k).
}
\end{lemma}
\begin{proof}
The lemma follows because we are using the \emph{max-max} model.
First, suppose that \math{u\in\gamma(x,k)}.
Nodes receive information in time steps.
Let \math{I_{m}(u)} be the information value at node \math{u} after time step
\math{m} of
information propagation
using a single seed \math{(x,k)}.
Similarly let \math{I'_{m}(u)} be the information value at \math{u}
after time step
\math{m} of
information propagation
using all sources and their respective seeds. We claim that
\math{I'_m(u)\ge I_m(u)} for all \math{m\ge0} and all \math{u}. Indeed,
suppose to the contrary that this does not hold for some \math{m,u};
in which case, there is an earliest time step \math{m} for which it
does not hold, i.e.
\math{I'_m(u)<I_m(u)} and for all \math{\ell<m} and all \math{v},
\math{I'_\ell(v)\ge I_\ell(v)}. It means that \math{I_m(u)} is different
from \math{I_{m-1}(u)}. Since there is no querying, and since
we are using the \emph{max-max} model, it means that some node
\math{v} \emph{propagated} information to node \math{u}
in the previous time step that resulted in information value
\math{I_m(u)=\a(v,u)I_{m-1}(v)}. But by assumption,
\math{I'_{m-1}(v)\ge I_{m-1}(v)} and since we are using the \emph{max-max}
model, \math{I'_m(u)\ge \a(v,u)I_{m-1}'(v)\ge\a(v,u)I_{m-1}(v)
= I_m(u)}, a contradiction. Thus, if a node is converted in any set $\gamma(x,k)$, then it is also converted in the joint information propagation using all the seed sets simultaneously.

Let $I^{\max}_m(u)$ denote the maximum value of $I_m(u)$ over all singleton
(seed,source) pairs \math{(x,k)} where  $x\in \p_k$ and $k\in\{1,\ldots, K\}$. We now claim that
\math{I'_m(u)\le I^{\max}_m(u)}.
This means that if a node gets converted in the joint information propagation,
it gets converted in at least one of the singleton information propagations.
Indeed, suppose to the contrary that for some earliest \math{m},
\math{I'_m(u)>I^{\max}_m(u)}.
Again this means that in the joint propagation, \math{I'_m(u)} changed
at step \math{m}, and so information was propagated from some node
\math{v} to \math{u} with the result that
\math{I'_m(u)=\a(v,u)I'_{m-1}(v)} (\emph{max-max} model).
We also know that there is some singleton propagation with
\math{I_{m-1}(v)\ge
I'_{m-1}(v)}, by assumption, since \math{m} is the earliest time step
when this fails. In this singleton propagation, it must be that
\math{I_{m}(u)\ge \a(v,u) I_{m-1}(v)
\ge \a(v,u) I'_{m-1}(v)=I'_m(u)} (\emph{max-max} model),
a contradiction.
\end{proof}

An immediate consequence of Lemma~\ref{lem:Gamma} is that if all the source
values are the same, so
$I_k=I$ for \math{k\in\{1,\ldots,K\}}, and every node trusts all
the sources equally, so \math{\a(k,u)} is independent of
\math{k}, then the converted set is
\mand{
\bigcup_{k=1}^K\bigcup_{u\in\p_k}\gamma(u,1).
}
This converted set is exactly what it would be if there was just one
source with value \math{I} seeding \math{\p=\cup_{k=1}^K\p_k}.
\begin{lemma}\label{lem:identical}
If all sources are identical and every node trusts all sources equally, then
\math{K} sources with information value \math{I} seeding sets
\math{\p_1,\ldots,\p_K} results in the same converted set as one source
with information value \math{I} seeding the set
\math{\p=\cup_{k=1}^K\p_k}.
\end{lemma}
It is therefore immediate that finding the optimal seed sets for
\math{K} identical sources with budgets
\math{B_1,\ldots,B_K} is equivalent to finding the optimal single seed set
for just a single one of these sources with budget
\math{\sum_{k=1}^KB_k}.
Thus, the \math{K} identical sources problem reduces to the single source
problem.

A direct application of Lemma~\ref{lem:Gamma}
gives the coverage function.
\begin{lemma}\label{lem:single} 
\mand{
\g(\p_1,\ldots,\p_K) = \left|\bigcup_{k=1}^K\bigcup_{u\in \p_k} \gamma(u,k)\right|.
}
\end{lemma}
It follows from Lemma~\ref{lem:single} that \math{\g(\p_1,\ldots,\p_K)} is monotone.
It is also easy to see that functions of this form (that are the size of the union taken nodewise of a set function defined on a node) are submodular.
We therefore have the following theorem.
\begin{theorem}\label{theorem:greedy}
For the \emph{max-max} model,
there is a greedy
deterministic algorithm which computes seed sets
\math{\p_1,\ldots,\p_K} of sizes \math{B_1,\ldots,B_K} for which
\mand{
\g(\p_1,\ldots,\p_K)\ge{\textstyle\frac12}\cdot\g(\p'_1,\ldots,\p'_K)
\qquad
\forall \p'_k \text{ with } |\p'_k|\leq B_k.
}
Moreover, when the individual budget constraints \math{|\p_k|\le B_k} are replaced by a total budget constraint
\math{|\p_1|+\cdots+|\p_K|\le B},
then the deterministic greedy
 algorithm yields seed sets $\p_1,\ldots,\p_K$
of total size equal to the budget $B$ such that
\mand{
\g(\p_1,\ldots,\p_K)
\ge\left(1-{\textstyle\frac1e}\right)\g(\p_1',\ldots,\p_K')
\qquad
\forall \p_k' \text{ with } |\p_1'|+\cdots+|\p_K'|\leq B.
}
\end{theorem}
\begin{proof}
The theorem follows from the monotonicity and submodularity of
\math{\g} implied by Lemma~\ref{lem:single} and the greedy deterministic
algorithm for
maximizing a monotone submodular function (\cite{Fisher1978}). The case with a total budget simply asks to find a set of (node,source) pairs
of  size at most $B$ maximizing a submodular monotone set function;
for this case the greedy algorithm gives a $1-\frac1e$ approximation. The case with multiple sources asks to maximize a submodular monotone set function with respect to partition matroid constraints; for this, the greedy algorithm from \cite{Fisher1978} gives a $\frac12$-approximation.
\end{proof}

\paragraph{Remark.}
When \math{K=1} (single source), the total budget constraint and the
individual budget constraints are equivalent and we get a
\math{(1-\frac1e)}-approximation.

\paragraph{Remark.}
Randomized algorithms guaranteeing a $(1-\frac1e)$-approximation exist for the case with multiple sources as well as for a single
source\cite{Calinescu2011}. We choose to focus on the greedy algorithms due to their simplicity and relative efficiency, as our goal is to find good seed sets for extremely large networks.

The greedy algorithm implied by Theorem~\ref{theorem:greedy} is an intuitive
algorithm that at each greedy step selects the node to add to the seed set
that gives the largest increase in the size of the converted set.
This algorithm
is summarized in Algorithm~\ref{alg:greedy} for the case with
individual budget constraints. When there is a total budget
constraint, the algorithm is exactly the same except that the
{\bf while} condition in step 4 checks that the total budget constraint is
not exceeded.
Our projected greedy heuristic for the general model is based on this
greedy algorithm. Implementation of this algorithm is discussed in Section~\ref{sec:greedy}.


\begin{algorithm}[t]\label{algo:greedy}
\begin{center}
\fbox
{
\parbox{5in}{
\begin{algorithmic}[1]
\STATE {\bf Input:} Instance of the \emph{max-max} model with $K$ sources
and budgets \math{B_1,\ldots,B_K}.
\STATE	Initiate converted nodes $C = \emptyset$ and selected nodes
\math{\p_k=\emptyset}.
\STATE Compute \math{C_{u,k}=\gamma(u,k)} for all nodes
\math{u\in V} and sources \math{k=1,\ldots,K};
\WHILE{$|\p_k| < B_k$ for any $k$ and $|C| < |V|$}
\STATE	Choose $(u^*,k^*)$ with $|\p_{k^*}|<B_{k^*}$ that maximizes
$|C_{u,k}\setminus C|$ over \math{(u,k)}\;
\STATE	Update $\p_k\gets \p_k\cup u^*$ and $ C = C\cup C_{u^*,k^*}$\;
\ENDWHILE
\STATE	Output sets $\p_1,\ldots,\p_K$.
\end{algorithmic}
}
}
\end{center}
\caption{Greedy algorithm on instance of simplified model\label{alg:greedy}}
\bigskip

\end{algorithm}

\section{The Projected Greedy Heuristic}\label{sec:greedy}

We now describe the Projected Greedy heuristic which takes as input
an instance
$\G$ of the general diffusion model (see Figure~\ref{fig:instance})
with \math{K} sources
 and produces as output a seed set $\{\p_{1},\p_{2},\hdots,\p_{K}\}$.
Given the seed sets $\{\p_{1},\p_{2},\hdots,\p_{K}\}$, it is possible to compute the coverage set
by simulating the information set updates in the network; simulating a single time-step takes \math{O(|E|)}.
We will denote the time it takes to compute the coverage set in the general model by $T_\G=O(a|E|)$, where $a$ is
the number of timesteps that it takes for the diffusion process to converge. While theoretically $a$ can be large
for the general model of diffusion we consider, practically we have observed in our experiments that 30-50 timesteps
is sufficient. Thus typically $T_\G$ is on the order of the number of links in the network, which for sparse social
network graphs is \math{O(|V|)}. This is useful, because all our algorithms need to be able to evaluate a seeding
in order to improve it. First, for comparison, we describe two natural algorithms because our projected
greedy algorithm has aspects of both.

\subsection{Brute Force}

The brute force algorithm is extremely simple:
try every possible distinct seeding, compute
\math{\g} for each seeding
 and select the seeding which maximizes the coverage. The running
time of this naive brute force approach is \math{O\left(T_\G\prod_{k=1}^K\choose{|V|}{B_k}\right)} because
there are \math{\prod_{k=1}^K\choose{|V|}{B_k}} possible seed sets.
Clearly, this algorithm
is not feasible, even for \math{B_k=1} (there are just too many possible
seedings to test). However, if we could somehow intelligently prune this
collection of seedings to a small number, then it would become
viable. This is one aspect of the projected greedy approach:
obtain plausible candidate seedings and pick one of them using the
brute force approach.

The way we will obtain these plausible candidate screenings is using a greedy
approach. The natural greedy algorithm is what we describe next, which
we call \emph{Actual Greedy}. Actual greedy deterministically
produces a single seed set.

\subsection{The Actual Greedy Approach}

The natural greedy seeding strategy is also simple.
For every pair \math{(u,k)} where \math{k} is a source and \math{u\in V} a
node, we consider adding it to the seed set (as long as we do not
violate the budget constraints \math{|\p_k|\le B_k}).
We add the pair \math{(u^*,k^*)} which results in the largest increase in the
coverage \math{\g(\p_1,\ldots,\p_K)}.
Starting with an
empty seeding, in this way we build up the seeding to the final
seeding one pair at a time.
We call this strategy {\em Actual Greedy}. While \emph{Actual Greedy} is
a plausible heuristic, because a general instance of the diffusion model is
non-monotone and non-submodular, there is no performance guarantee
as compared with using the greedy algorithm for the simplified model
(Section~\ref{section:simplified}).

\paragraph{Running Time of \emph{Actual Greedy}.}
Each step of \emph{Actual Greedy} adds a node to the seeding.
If the total budget is
\math{B=\sum_{k=1}^KB_k}, then there are \math{B} steps in the algorithm.
At each step, we need to consider \math{O(K|V|)} pairs to determine
which one is the best to add; and the time to test the seeding with that
pair added is
\math{T_\G}. So the total running time is
\math{O(B\cdot K\cdot|V|\cdot T_\G)}. For sparse graphs,
\math{|E|=O(|V|)}, and when the total size of the seed set is
a fraction of the graph (\math{B=\Theta(|V|)}), this means that the running time is typically
\math{O(K|V|^3)}.

The general diffusion model for arbitrary settings of the parameters does not
have any nice properties that can be exploited to improve this
running time, and so for large graphs, with millions of nodes, this
cubic running time is practically infeasible.

\subsection{Projected Greedy}

Given a general instance
of the diffusion model \math{\G}, the idea behind the
\emph{Projected Greedy} algorithm is to construct an instance of the
simplified model \math{\S} (see Section~\ref{section:simplified})
that closely approximates \math{\G}.
In a matter of speaking, we are ``projecting'' \math{\G} down
to the simpler instance \math{\S} that most ``closely'' approximates
\math{\G}.
For the simpler instance \math{\S}, we can leverage
Theorem~\ref{theorem:greedy} and use the greedy algorithm to
obtain a constant factor approximation to the optimal seeding
in \math{\S}. If the instance \math{\S} is a close approximation to
the instance
\math{\G} then the near optimal seeding for
\math{\S} will also be a good seeding for~\math{\G}.

There are two advantages of \emph{Projected Greedy} over
\emph{Actual Greedy}: (1) Since the simplified
model is monotone and submodular, we can use a more efficient algorithm than
running the \math{O(K|V|^3)} greedy algorithm on the general instance
\math{\G}; (2) Neither \emph{Actual Greedy} nor \emph{Projected Greedy} give
any performance guarantee for the quality of the seeding. However, we do
have some flexibility in the choice of the simplified instance
\math{\S}. So, by exploring a variety of plausible
simplified instances, we can generate several seedings (say \math{c} of
them), and we can choose one of these \math{c} as in the Brute Force
approach, via simulation. One can view this as a more intelligent sampling
of the possible seedings to test in the Brute Force approach, where the
choice of seedings is guided by the various choices for the simplified
instances \math{\S}.

Although we cannot give a guarantee on
the quality of the seeding produced by our heuristic, our experimental
evaluation demonstrates that \emph{Projected Greedy} performs significantly
better than widely used simple seeding strategies.

\subsection{Creating an Instance of the Simplified Model}

Given an instance \math{\G} of the general diffusion model
(see Figure~\ref{fig:instance}), we construct an instance
\math{\S}
of the
simplified model by setting \math{\lambda_d=\lambda_s=0}: an instance
of the \emph{max-max} model, and we set the evacuation
time \math{\tau} to be sufficiently large so that there is no evacuation.
We set the upper and lower thresholds to the same value,
\math{t_l(u)=t_h(u)=t(u)} so that there is no Undecided state.
We do not change any of the trust values along any of the edges, or between
sources and nodes.

When \math{\lambda_s>0} or \math{\lambda_d>0} the diffusion will generally
be faster because taking the sum will tend to
inflate information values at the nodes. To compensate for this,
we need to raise the thresholds, to get a better approximation to \math{\G}.
We treat \math{t(u)} as tunable parameters in the simplified model, and
we may alter the values to get different instances of the
simplified model. We will discuss how to choose these thresholds shortly.
An instance of the simplified model is summarized
in Figure~\ref{fig:SimpInstance}.
\begin{figure}
\begin{center}
\fbox{
\parbox{4.55in}{
{\bf Instance \math{\S}}
\\
Graph \math{G=(V,E)} specifying the network for the diffusion.
\\
Source information values and budgets,
\math{(I_1,B_1),\ldots,(I_K,B_K)}
\\
Trust values \math{\a(u_i,u_j)} for all edges \math{(u_i,u_j)\in E}.
\\
Trust values \math{\a(k,u_j)} between each source \math{k} and each node
\math{u_j\in V}.
\\
Diffusion parameters \math{\lambda_d=\lambda_s=0}, \math{\tau\rightarrow
\infty}
\\
Lower and upper thresholds
\math{t_l(u)=t_h(u)=t(u)} for each node \math{u\in V}.
\\
{\bf Desired output:} seed sets \math{\p_1,\ldots,\p_K}.
}}
\end{center}
\caption{The simplified instance \math{\S} for general instance
\math{\G}.\label{fig:SimpInstance}}
\end{figure}
We distinguish quantities in the simplified model using a subscript
\math{\S} from quantities in the general model with a subscript \math{\G}; for
example,
$\g_{\G}(\p_1,\ldots,\p_K)$ is the coverage function for a seeding
in the general instance, and
$\g_{\S}(\p_1,\ldots,\p_K)$
is the coverage function for the same seeding in the
simplified instance.

\paragraph{Further Simplifying the Model}
We can further simplify the model by insisting on a single source with
information value \math{I} and budget equal to the total budget of
the \math{K} sources, \math{B_\S=B_1+\cdots+B_K}. The
information value \math{I}
is chosen as the
weighted information value of the \math{K} sources, where the weights are
the number of seeds allocated to each source:
$$ I_\S = \frac{1}{\sum\limits_{k = 1}^{K} B_{k}}\cdot
\sum_{k = 1}^{K}I_{k}\cdot B_{k}.$$
With \math{K} sources, every node \math{u} has a trust weight
\math{\a(k,u)} for each source. These are combined into a single
average trust weight with a single source:
\mand{
\a_\S(1,u)=\frac{1}{K}\sum_{i = k}^{K}\a_\G(k,u).
}
All trust values between two nodes in the graph are unchanged.
We will discuss next how to select the thresholds
\math{t_\S(u)}. By considering different thresholds at the nodes, we
are able to generate a variety of seedings, each of which is near optimal
for a slightly different instance of the simplified model.

In this further simplified model, a single source will seed
\math{B_\S} nodes giving a seeding
\math{\p_\S}. To convert it into a seeding for the general instance
\math{\G} we need to assign each seeded node to one of the
\math{K} sources in \math{\G}, so \math{\p_\S\rightarrow
\p_1,\ldots,\p_K}. In principle, one could try to optimize this
assignment. For simplicity, we just pick an arbitrary assignment
(for example a random assignment) of the seeded vertices to the
\math{K} sources that respects the constraint \math{|\p_k|\le B_k}.

\subsection{Using Thresholds \math{t_\S(u)} to Generate Different Seedings}

For different choices of
\math{t_\S(u)} with \math{u\in V}, we get different instances of the
simplified model:
\mand{
\{t_\S(u)\}\quad
{\buildrel\text{\footnotesize greedy}\over \longrightarrow}
\quad \p_\S\quad
{\buildrel\text{\footnotesize partition}\over \longrightarrow}
\quad \p_1,\ldots,\p_K\quad
{\buildrel\text{\footnotesize evaluate}\over \longrightarrow}
\quad
\g_\G(\p_1,\ldots,\p_K).
}
By evaluating according to \math{\g_\G} in this way, we can use different
settings of the thresholds in the simplified model to explore different
candidate seedings for the general instance \math{\G} in a more intelligent
way than the pure Brute Force approach.

There are several ways to choose the thresholds in the simplified model
to get different seedings for the general instance. For simplicity and
computational efficiency, we assume homogeneous thresholds, so every node
has the same upper and lower threshold, and so \math{t_\S(u)} becomes
just \math{t_\S}. We choose \math{t_\S} from a set of
thresholds
$\Omega = \{t_{1},t_{2},\hdots,t_{c}\}$
where $0\leq t_{i} \leq 1$ to generate \math{c} simplified instances;
this in turn will produce \math{c} seedings to be evaluated with
\math{\g_\G}. We choose \math{\Omega} and \math{c} as follows.

\paragraph{Homogeneous Approximation to Trust in Instance \math{\G}.}
To construct \math{\Omega}, we imagine
all trust weights in instance \math{\G} are
\math{\a_{avg}}, the average trust weight in the network for instance
\math{\G} (this is just for the intuition on how we generate \math{\Omega}).
Let $t_{min}^\G$ be the minimum lower threshold over all nodes in the
general instance
\math{\G}; similarly $t_{max}^\G$ is the maximum upper threshold.
All the thresholds in \math{\Omega} will satisfy
\math{t_{min}\le t\le t_{max}}.
The highest threshold we need to consider is \math{\a_{avg}\cdot I}
because no information value at a source can be higher than $I$, so if the
threshold is higher, no node can convert to Believer. Every next
hop reduces the information value by a factor of
\math{\a_{avg}} so all thresholds which are in the interval
\math{[\a_{avg}^2\cdot I\ ,\ \a_{avg}\cdot I)}
are equivalent, and we only need consider
one of them, \math{\a_{avg}^2\cdot I}. We can continue this logic which
implies that we only need to consider thresholds (in the simplified model) of
the form
\math{t_i=\a_{avg}^i\cdot I}.
Further restricting the thresholds to  be in the range
\math{[t_{min},t_{max}]}
results in our choice of the threshold set \math{\Omega}:
\mand{
\Omega=
\left\{t_i\ \left|\ t_i=\a_{avg}^i\cdot I;\  t_{min}\le t_i\le
t_{max}\right.\right\}
\cup
\{t_{min},t_{max}\}.
}
The number of thresholds \math{c\approx 2+O(\log t_{min}/\log\a_{avg})}.

\paragraph{Two-Level Approximation to Trust in Instance \math{\G}.}
A better approximation to the trust values in \math{\G} can be obtained
by clustering the trust weights into high values \math{\a_{high}} and
low values \math{\a_{low}}
This leads to a larger set of thresholds in \math{\Omega} for the
simplified model. The general idea is the same. We choose thresholds
as the possible information values at the nodes.
This set of information values is
\mand{
\{I;\ \a_{high}\cdot I,\ \a_{low}\cdot I;\
\a_{high}^2\cdot I,\ \a_{high}\a_{low}\cdot I,\ \a_{low}^2\cdot I;\
\ldots\}
}
In general, we can construct the set \math{\Omega} iteratively as follows:
\begin{center}
\begin{algorithmic}[1]
\STATE\math{\Omega\gets\{I\}}.
\WHILE{\math{\a_{high}\cdot\max_{t\in\Omega} t\ge t_{low}}}
\STATE \math{\Omega_{add}\gets\emptyset}.
\FOR{every \math{t\in\Omega} with \math{\a_{high}\cdot t\ge t_{min}}}
\STATE add \math{t} to \math{\Omega_{add}}.
\ENDFOR
\FOR{every \math{t\in\Omega} with \math{\a_{low}\cdot t\ge t_{min}}}
\STATE add \math{t} to \math{\Omega_{add}}.
\ENDFOR
\STATE\math{\Omega\gets\Omega\cup\Omega_{add}}
\ENDWHILE
\end{algorithmic}
\end{center}
This algorithm easily generalizes to approximating the trust in the instance
\math{\G} by more than two trust levels.
The size \math{c}
of the resulting set of thresholds \math{\Omega} depends on
how quickly information decays in the network; faster information decay
leads to smaller \math{c}.

\subsection{Running Time of Projected Greedy}
\label{sec:rumtime}
The basic methodology of \emph{Projected Greedy}
includes constructing \math{c} instances of the simplified model and
computing the near-optimal seeds in the simplified instance
\math{\S}, converting that solution to a seeding for the general instance
\math{\G} and finally evaluating \math{\g_\G} for the seeding.
The running time is therefore
\mand{
c\cdot O\left(T_{greedy}+T_\G\right),
}
where \math{T_{greedy}} is the time to construct the solution for the
simplified instance using the greedy algorithm.
Our implementation of the greedy algorithm for the simplified instance
exploits
the monotone and submodular nature of
\math{\g_\S}.
We create a special data structure $\delta$
where $\delta(u)\subseteq V$ is the set of nodes that convert
\math{u} when they alone are added to the seed set.
While calculating $\delta$, we also create array $N$
where $N(u)$ is the number of nodes that $u$ (if added to the seed set)
converts to the Believed state.
For sparse graphs, building $\delta$ and $N$ takes $O(n^{2}\log n)$ time, where $n=|V|$.
After selecting the node that gives the best improvement in $\g_{\S}$,
we only need to update the array $N$ using $\delta$.
This means we do not need to recalculate $N$ at every step,
 and the update takes only $O(n)$ time.
The process of updating array $N$ depends on properties of $\g_{\S}$,
such as monotonicity. Thus, for sparse graphs,
\mand{
T_{greedy}=O(n^2\log n+Bn).
}
When \math{B=\Theta(n)}, on sparse graphs,
the total running time is typically
\math{c\cdot O(n^2\log n)}.
To give an idea about how large the value of $c$ is, consider
an extreme case with homogenous trust values where $t_{min} = 0.01$,
$t_{max}= 0.99$ and $\a_{avg} = 0.9$.
In this extreme scenario $c = 44$,
which gives a running time that is
 orders of magnitude  better than \emph{Actual Greedy} when
\math{n} is in the millions.

Pre-calculating the data structure $\delta$ can have
worst case space complexity of $O(n^{2})$.
To avoid this, we use a hybrid strategy for computing seed sets in
\emph{Projected Greedy}.
In this strategy $\delta$ is not stored initially and the 
array $N$ is recalculated each time the best node is to be selected.
We use a technique in \cite{Krause2008} according to which array $N$ is stored
as a priority queue and each selection requires recalculation of $N(u)$
for only a small number of nodes $u$.
Only after a certain threshold is breached do we switch to populating $\delta$
for all nodes.
Since nodes that convert the highest number of other nodes to Believed state have already been selected at this stage, it helps in reducing the size of $\delta(u)$ for each of the remaining nodes.
Not only does this help reduce memory usage but it also does not affect the running time adversely.
Ultimately, \emph{Projected Greedy}  is more efficient by a factor of
about \math{n} compared to
\emph{Actual Greedy}.
In large graphs, this is a significant improvement and can sometimes be the difference between feasibility and intractability.

\section{Experiment Design}
\label{sec:design}

In order to compare the Greedy heuristic with other seeding strategies, we simulate the spread of evacuation warnings in a social network. The simulation of diffusion is carried out on different types of networks with several parameter values.

\subsection{Networks}
We used three different network structures, each with $100,000$ nodes with average degree approximately
 $4$. Frequently, in real world scenarios individuals form groups based on race, ethnicity, nationality, etc. Not only are individuals within the same social group more likely to be acquainted with each other, they are also inclined to place more trust with people in the same group as theirs. Since these connections and trust disparities play an important role in the flow of information, we model social groups by dividing the population into $2$ groups in the
following networks.

\paragraph{Scale-free Graphs:} We use the Albert-Barabasi model for generating random scale-free networks using the preferential attachment model \cite{Barabasi1999a,Albert2002}. The graph thus produced follows a 
power law distribution with exponent approximately $-2.9$, that is
 $P(x)\propto  x^{-2.9}$ (\math{P(x)} is the fraction of nodes having
degree \math{x}).
 Once the graph is created, $50,000$ nodes are randomly assigned to one group and the rest to the other.
Scale-free graphs do not take into consideration the fact that nodes within a group are more likely to communicate with each other.

\paragraph{Random Group Model:} Nodes are randomly assigned into $2$
groups of size
$50,000$ nodes. The probability that $2$ nodes are connected (edge probability) depends on whether they belong to the same group or not. If two nodes belong to the same group then the edge probability is $p_{s}$ while if they belong to different groups, the edge probability is $p_{d}$. Here $p_{s} = 2 * p_{d}$
(more connections within a group as between groups).
The probability \math{p_d} is chosen
so that the average degree is $4$.

\paragraph{San Diego Network:} This is a random geometric graph that is constructed from actual demographic data in the San Diego area \cite{Hui2010,Hui2011}. Since there is a large population of Hispanics, we consider two groups:
Hispanics and Non-hispanics. The population of nodes belonging to each group is based on the racial demographics of the region. Also, the edge probability between two nodes depends not only on the group they belong to, but also the physical distance between them. For example the edge probability between two nodes belonging to the same group is larger if they live close to each other than if they live far apart. Again the edge probabilities are chosen to have average
degree~$4$. The details of this model are given in
\cite{Hui2010,Hui2011}.

\subsection{Node Characteristics}
Since evacuation is a high cost action,
nodes would not evacuate without a significant amount of information.
This can be modeled with high upper thresholds $t_{h}(u)$.
Also, with such a high risk situation as evacuation,
individuals may be proactive. That is to say,
they may be willing to put more effort in collecting information.
This can be modeled by decreasing the lower thresholds $t_{l}(u)$.

In our experiments, we simulate $3$ different threshold value pairs for all nodes: $(t_{l}=0.2,t_{h}= 0.3)$, $(t_{l}=0.15,t_{h}= 0.55)$,
and $(t_{l}=0.4,t_{h}= 0.5)$.
We set the evacuation time to
 $\tau = 5$ time-steps (so, all nodes leave the network
5 time steps  after they are converted to their Believed state).

\subsection{Trust Scenarios}
Since nodes are split into $2$ groups, there are $2$ kinds of edges in the graph. The first type of edge is incident on nodes from the same group
(denoted type $A$ edges) . The second type of edge is incident on nodes from different groups  (denoted type $B$ edges).
Based on the trust values on these edges we have two trust scenarios.
In each scenario, we set the average trust on the edges to be $\a$.

\paragraph{Homogenous trust:} All edges have the same trust value. This models situations when no social groups exist.
The trust value on every edge is $\a$.

\paragraph{Group Variable trust:} The trust value on type $A$ edges is
$\a + \varepsilon$ where $\varepsilon > 0$. The trust value on type
$B$ edges is chosen so that the average trust is $\a$.
So the trust on type $B$ edges will be less than type
\math{A} edges, which models social groups that are more trusting of their
own group than outsiders.
We used $\a = 0.7$ and $\varepsilon = 0.05$ for our simulations.

\subsection{Seeding Algorithms}
We have $5$ trustworthy sources each with information value $I  = 0.95$
and trust value $\a(k,u) = 0.9$ for all $u \in V$. We look at scenarios in which between $5\% - 50\%$ of nodes are seeded in total, with
each source seeding an equal number of nodes.
We compare the following algorithms for generating the seeding
of total size $B$.

\paragraph{Random:} Randomly select \math{B} nodes
and arbitrarily assign these nodes to the \math{K} sources.

\paragraph{High Degree:} Select the $B$ highest degree nodes and arbitrarily
assign them to the \math{K} sources. Here, the 
degree for node $u$ is the
total outgoing trust weight:
$$degree(u)= \sum_{(u,u')\in E} \a(u,u')$$

\paragraph{Projected Greedy heuristic:} Seeds are generated according to the Projected Greedy heuristic described in Section $\ref{sec:greedy}$.

\subsection{Parameters}
For all our experiments we use $\lambda_{s} = 0$. This means nodes always chose the maximum value while combining information from the same source. We ran
 simulations for different choices of
$\lambda_{d}\in\{ 0.0, 0.05, 0.1, 0.2\}$.

Lastly, in real life scenarios, communication between nodes
 is not likely to succeed every time.
Hence, when a node that is not a source, queries or propagates an
information set to another node, it will succeed with some
probability \math{p}; we set  $p= 0.75$ in our experiments.

\section{Results and Discussion}
\label{sec:disc}

We ran each simulation for $50$ steps and repeat it $100$ times. We
 observed that $50$ steps are enough for the diffusion process to conclude.
Since the networks are generated randomly, we repeat the whole simulation on at least $10$ instances of the graph.
We use the 
average number of nodes evacuated as our measure of performance.
The standard deviation (error bar) due to the randomness
is extremely small; in fact the error bars are not even visible in our
plots.

\subsection{Seed Size}

\begin{figure}[t]
\begin{center}
\includegraphics[width = 6.5in]{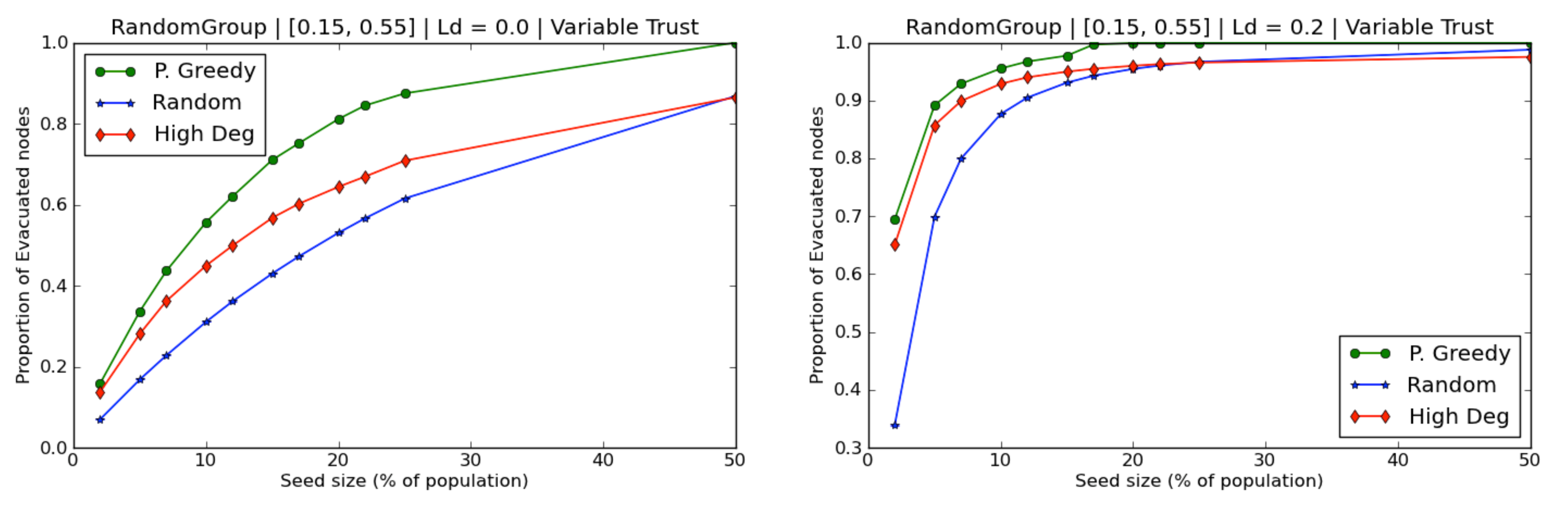}
\caption{\small
Random Group network with group variable trust (high trust
edges within a group). Thresholds are $t_l=0.15,\ t_h=0.55$; We
show
 $\lambda_{d}= 0.0$ (left) versus $\lambda_{d}=0.2$ (right) for
various total seed budget. Note that for plot on the right, the $y$-axis range is $[0.0, 1.0]$ while for the right plot it is $[0.3, 1.0]$. The Projected Greedy dominates the other algorithms.}
\label{fig:randomgroup}
\end{center}
\end{figure}
We start with the performance of the three seeding strategies (random, high-degree and Projected Greedy) as the total seeding budget increases.
Figure~\ref{fig:randomgroup} shows the random group model and
Figure~\ref{fig:sandiego} the San Diego network.
Though we ran several different scenarios, we pick these two as both relevant and
representative of the typical nature of the results.

For both the random graph model and the San Diego network,
the Projected Greedy heuristic performs consistently better
than both Random and High Degree seeding strategies.
When $\lambda_{d} = 0.0$, Projected Greedy and High Degree
are comparable for a small seed budget, but
as the seeding budget increases, so does the gap between their performance.
This is likely because there are two competing effects:
one should seed influential nodes, and one should spread out the seeds (i.e., not have the seed nodes be too close to each other, so that they will convert more nodes in total).
When you have a few seeds, influence is more important, and Projected Greedy
and high-degree are achieving that goal comparably. When you have many seeds, it now becomes more important to spread the seeds out, and high-degree
ignores how spread out the seeds are, whereas
Projected Greedy takes that into account.
This also explains why, eventually, even random seeding becomes better than
high degree: with more seeds, the need to spread out
the seeds wins. This fact is even more pronounced when
\math{\lambda_d=0.2} which makes sense because many spread out seeds
can have even more impact due to the summing effect present when
\math{\lambda_d>0}.
\begin{figure}[t]
\begin{center}
\includegraphics[width = 6.5in]{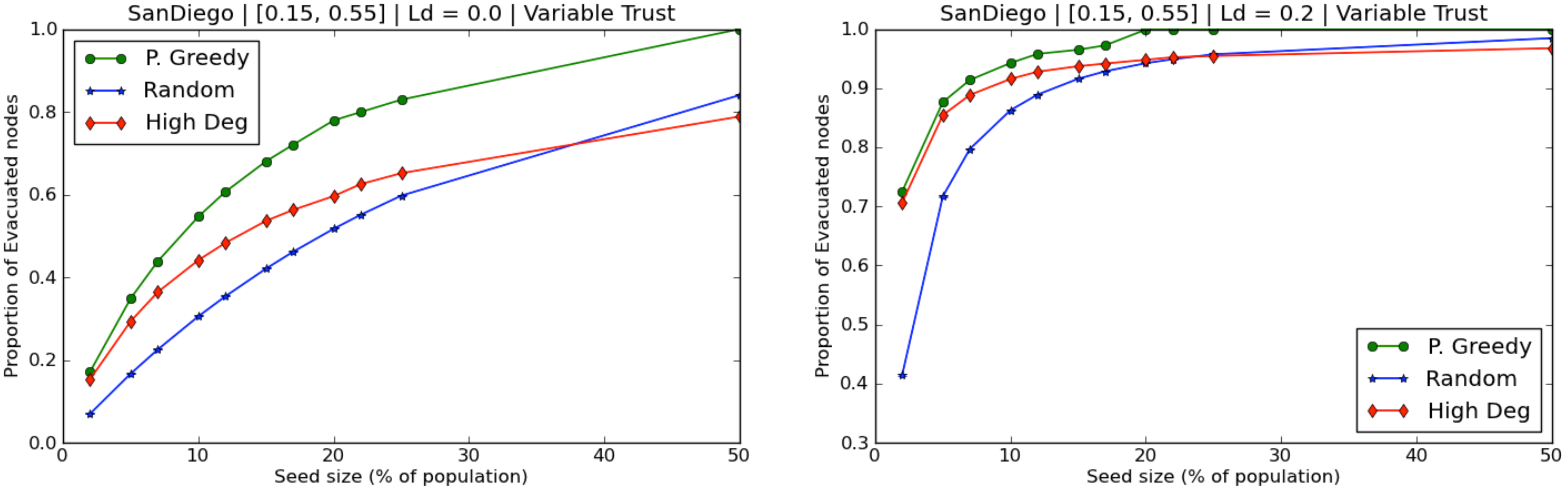}
\caption{\small San Diego network, with group variable trust (high trust
edges within a group). Thresholds are $t_l=0.15,\ t_h=0.55$; We
show
 $\lambda_{d}= 0.0$ (left) versus $\lambda_{d}=0.2$ (right) for
various total seed budget.
Note that for plot on the right, the $y$-axis range is $[0.0, 1.0]$
while for the right plot it is $[0.3, 1.0]$. Projected Greedy clearly dominates
the other seeding strategies.}
\label{fig:sandiego}
\end{center}
\end{figure}
This effect is even stronger in the San Diego network
(Figure~\ref{fig:sandiego}),
probably due to the structure of San Diego network. Such typical social networks
have a
 number of small dense clusters that weakly connect to each other.
This makes selecting high degree nodes from the same cluster a bad idea.

The comparative advantage of Projected Greedy is higher with
\math{\lambda_d=0} for two reasons. The first is that the
simplified instance \math{\S}
on which the seed set is near optimal is a
 better approximation to the true instance \math{\G} when
\math{\lambda_d=0}. Second, when you have some element of summing in the
diffusion process, all algorithms will improve, and hence their differences
will diminish. Nevertheless, Projected Greedy still outperforms
even for
$\lambda_{d} = 0.2$, managing to convert almost every node
to believed state with seed set size as low as $20\%$; even at 50\% seeds, the
other algorithms cannot achieve this.

\subsection{Cross section of Results}

We give a comprehensive cross section of the results in
Table~\ref{tab:results}
for all three types of networks, both homogeneous and group-variable trust,
and for a variety of trust thresholds.
\begin{table}[t]
\begin{center}
\begin{tabular}{l l l c c c c c c }
	\toprule
		& & & \multicolumn{3}{c}{Homogenous Trust} & \multicolumn{3}{c}{Variable Trust}\\
		\cmidrule (lr{.75em}){4-6} \cmidrule (lr{.75em}){7-9}
		& & & R & HD & PG & R & HD & PG \\
		\midrule
\multirow{5}{*}{SF} & \multirow{2}{*}{$t_l=0.2,\ t_h=0.3$} & $\lambda_{d}:0.0$ & $25.76$ & $1.33$ & $\mathbf{0.00}$ & $4.32$ & $0.08$ & $\mathbf{0.00}$ \\
\cmidrule {3-9}
& & $\lambda_{d}:0.2$ & $\mathbf{0.00}$ & $\mathbf{0.00}$ & $\mathbf{0.00}$ & $\mathbf{0.00}$ & $\mathbf{0.00}$ & $\mathbf{0.00}$ \\
\cmidrule {2-9}
& \multirow{2}{*}{$t_l=0.15,\ t_h=0.55$}  & $\lambda_{d}:0.0$ & $70.76$ & $3.23$ & $\mathbf{0.00}$ & $71.00$ & $3.20$ & $\mathbf{0.00}$ \\
\cmidrule {3-9}
 & & $\lambda_{d}:0.2$ & $25.75$ & $1.33$ & $\mathbf{0.00}$ & $14.39$ & $0.63$ & $\mathbf{0.00}$ \\
\midrule
\multirow{5}{*}{RG} & \multirow{2}{*}{$t_l=0.2,\ t_h=0.3$} & $\lambda_{d}:0.0$ & $31.85$ & $11.16$ & $\mathbf{0.00}$ & $16.21$ & $3.87$ & $\mathbf{0.00}$ \\
\cmidrule {3-9}
& & $\lambda_{d}:0.2$ & $2.40$ & $2.03$ & $\mathbf{0.00}$ & $2.77$ & $2.30$ & $\mathbf{0.00}$ \\
\cmidrule {2-9}
& \multirow{2}{*}{$t_l=0.15,\ t_h=0.55$}  &$\lambda_{d}:0.0$ & $46.46$ & $10.94$ & $\mathbf{0.00}$ & $50.06$ & $16.48$ & $\mathbf{0.00}$ \\
\cmidrule {3-9}
& & $\lambda_{d}:0.2$ & $32.20$ & $10.64$ & $\mathbf{0.00}$ & $21.68$ & $3.87$ & $\mathbf{0.00}$ \\
\midrule
\multirow{5}{*}{SD} & \multirow{2}{*}{$t_l=0.2,\ t_h=0.3$} & $\lambda_{d}:0.0$ & $31.85$ & $11.02$ & $\mathbf{0.00}$ & $14.09$ & $2.91$ & $\mathbf{0.00}$ \\
\cmidrule {3-9}
& & $\lambda_{d}:0.2$ & $2.68$ & $2.33$ & $\mathbf{0.00}$ & $3.32$ & $2.76$ & $\mathbf{0.00}$ \\
\cmidrule {2-9}
& \multirow{2}{*}{$t_l=0.15,\ t_h=0.55$}  &$\lambda_{d}:0.0$ & $47.28$ & $11.39$ & $\mathbf{0.00}$ & $52.14$ & $15.60$ & $\mathbf{0.00}$ \\
\cmidrule {3-9}
& & $\lambda_{d}:0.2$ & $32.19$ & $10.52$ & $\mathbf{0.00}$ & $18.16$ & $2.41$ & $\mathbf{0.00}$ \\
\bottomrule
\end{tabular}
\end{center}
\caption{\small Cross section of results. The table shows
the relative regret of an algorithm for a particular scenario, a regret of
\math{\mathbf{0}} indicating the best performing algorithm. The Project Greedy
heuristic significantly dominates all the other algorithms in almost
every scenario. SF=scale free network; RG=random group model; SD=San Diego network.\label{tab:results}}
\end{table}
We quantify the performance in terms of the regret: for a given scenario,
the best performing algorithm has regret zero, and otherwise,
$$ regret = \frac{\g_{best}-\g}{\g_{best}}\times 100\%,$$
where $\g_{best}$ is the number of nodes evacuated by
the best algorithm for that
scenario, and $\g$ is the number of nodes evacuated by the algorithm
whose regret is being computed. A good algorithm would always have low regret.
All results are for a total seeding size of $5\%$.
The actual fraction of the network evacuated is shown in
 Table~\ref{tab:sf} for a particular instance.

In general, we make the following observations. When the number of
evacuated nodes
is large, which occurs either with low thresholds or a high value for
\math{\lambda_d} (more information aggregation),
the performance edge delivered by Projected Greedy is smallest.
In fact when \math{\lambda_d} increases beyond $ 0.2$, the effect of
 information aggregation quickly takes over and there is little difference between the performance of the three algorithms as almost all nodes
are evacuated with ease.
When the thresholds are high or \math{\lambda_d\approx0} the performance
edge delivered by Projected Greedy is quite significant,
in accordance to results shown in Figures \ref{fig:randomgroup} and
 \ref{fig:sandiego}. Projected Greedy's advantage increases with more seeds.

Note that in scale free graphs, high-degree is comparable to Projected Greedy
but random significantly underperforms. This is because in such networks
there are a few extremely important nodes, and it is essential to include
these nodes in the seed set, which is unlikely to happen with random
seeding.

\begin{table}[t]
\begin{center}
\begin{tabular}{l l c c c c c c }
	\toprule
		& & \multicolumn{3}{c}{Homogenous Trust} & \multicolumn{3}{c}{Variable Trust}\\
		\cmidrule (lr{.75em}){3-5} \cmidrule (lr{.75em}){6-8}
		& & R & HD & PG & R & HD & PG \\
		\midrule
\multirow{4}{*}{$t_l=0.15,\ t_h=0.55$ $\lambda_{d}:0.0$} & SF & $20.82$ & $68.89$ & ${71.19}$ & $20.65$ & $68.92$ & ${71.20}$ \\
\cmidrule{2-8}
& RG & $22.26$ & $37.03$ & ${41.58}$ & $16.85$ & $28.17$ & ${33.73}$ \\
\cmidrule{2-8}
& SD & $22.19$ & $37.30$ & ${42.09}$ & $16.72$ & $29.48$ & ${34.93}$ \\
\bottomrule
		
\end{tabular}
\end{center}
\caption{\small
Fraction of the network evacuated for a particular scenario. Random
has roughly the same performance on all networks and is a
significant under performer on scale free networks.\label{tab:sf}}
\end{table}

\subsection{Threshold Selection}\label{sec:threshold}
As described in Section \ref{sec:greedy}, in order to select the best seed set we carry out simulations over a range of values for threshold $t$ in the simplified model.
Specifically we repeat the simulation $c$ times.
It is interesting to see how the coverage
\math{\g_\G} changes as we change the thresholds in the simplified
model \math{\S}. In particular, how the best threshold
\math{t_{opt}} in the simplified model (i.e. the threshold that gives the
closest approximation to \math{\G}) depends on the
 parameters of the general model, in particular
$(\lambda_{s},\lambda_{d})$,  the lower and upper node thresholds
$(t_{l},t_{u})$ and transmission probability $p$.
\begin{figure}[t]
\begin{center}
\includegraphics[width = 6.5in]{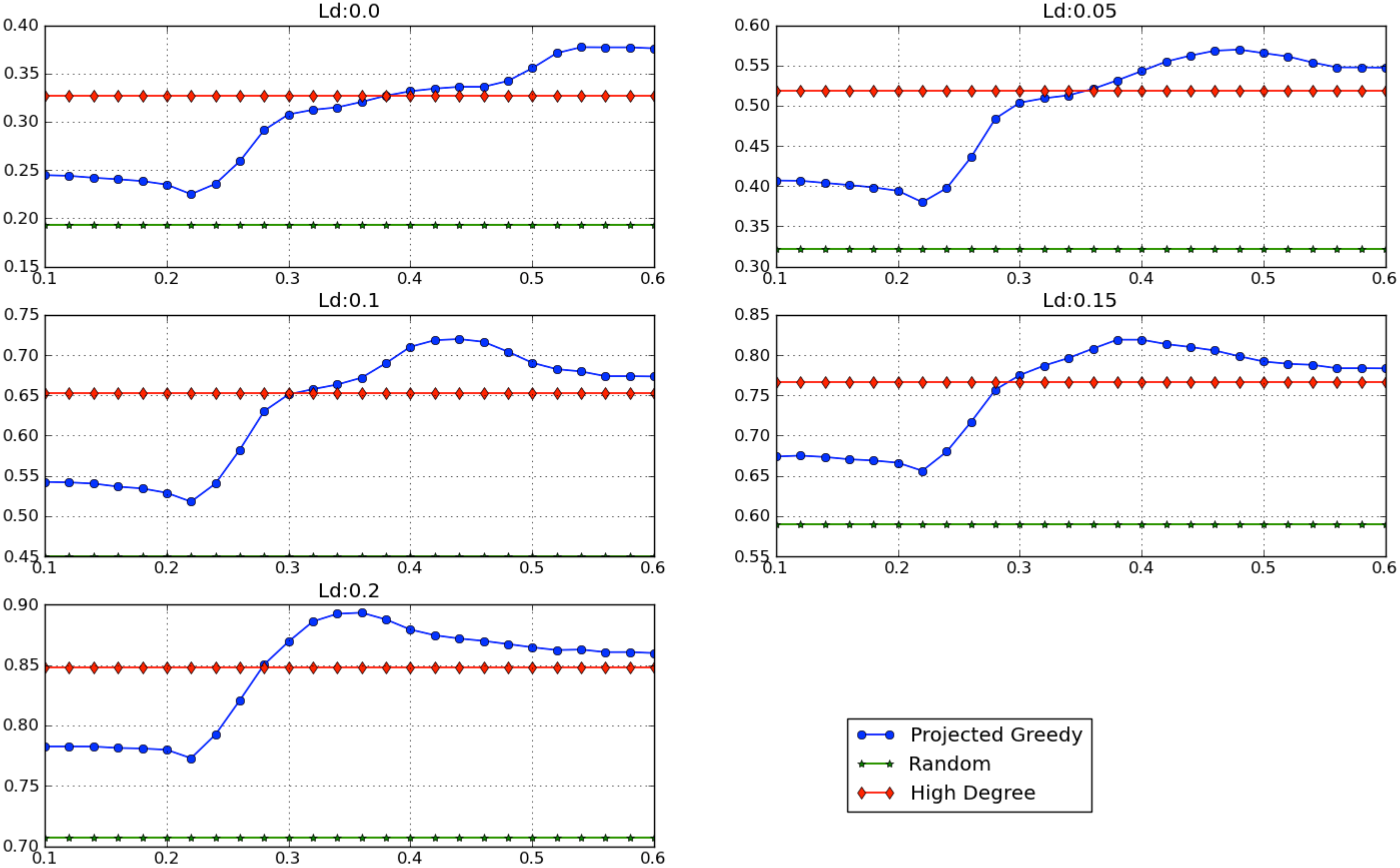}
\caption{\small
Plots showing change in the coverage in the general instance
\math{\G} obtained from using different node thresholds in the simplified
instance \math{\S}. We are interested in how the best threshold
in the simplified instance \math{\S} changes (the threshold that gives
the maximum coverage in \math{\G}) as we increase $\lambda_{d}$.
The $x$-axis shows the threshold in the simplified instance \math{\S} and the
$y$-axis shows the proportion of nodes evacuated in the general input
instance \math{\G}.
The plots are labeled with the value of $\lambda_{d}$.}
\label{fig:thrSelect}
\end{center}
\end{figure}

Our simulation results show that there is an interesting relationship between $t_{opt}$ (the best threshold to choose in \math{\S}
and $t_{u}$ (the upper thresholds in \math{\G})
under different values of $\lambda_{d}$.
In order to observe this relationship, we perform simulations on a generalized version of the San Diego network.
In this generalized version, edge trust values and node thresholds are selected uniformly at random from a range instead of being fixed to specific values.
Thus the generalized network tries to incorporate variations observed in real networks.
The network used for simulation has the following parameter values.
For edges between nodes belonging to the same group, trust values are selected uniformly at random from the range $[0.7,0.8]$.
For edges between nodes belonging to different groups, trust values are selected uniformly at random from the range $[\a_{low}-0.5, \a_{low}+0.5]$.
Here $a_{low}$ is selected such that the expected value of edge trusts in the network is $0.7$.
Similarly, for every $u\in V$, $t_{l}(u)$ is selected uniformly at random from range $[0.1,0.2]$ and $t_{u}(u)$ from range $[0.5,0.6]$.
For the probability of transmission $p = 0.75$ and seed set size $5\%$ of total nodes, Figure \ref{fig:thrSelect} shows the proportion of nodes evacuated with threshold values for Projected greedy in range $[0.1,0.6]$ and $\lambda_{d} = [0.0,0.05,0.1,0.15,0.2]$.

It is intuitive to expect that $t_{opt}$ will be close to $E[t_{u}]$ since nodes get converted to Believed state only after information value crosses $t_{u}$.
This is exactly what we see in the case where $\lambda_{d} = 0.0$.
The optimal threshold $t_{opt}$ is very close to $E[t_{u}] = 0.55$.
But as $\lambda_{d}$ increases we see a gradual decrease in the value of $t_{opt}$.
As $\lambda_{d}$ increases, information gets aggregated and the fused information value becomes progressively more successful at breaching $t_{u}$.
In other words, the general model starts behaving like a simple model with a smaller $t_{u}$ value where, for the same amount of initial information, it is comparatively easier to convert nodes.
Table \ref{tab:thr} shows how the value of $t_{opt}$ reduces with increasing $\lambda_{d}$.
The rate at which $t_{opt}$ decreases may depend upon factors such as the type of network (i.e. network structure), $t_{l}$ etc. These are interesting
questions for future work.

\begin{table}[t!]
\caption{$t_{opt}$ decreases as $\lambda_{d}$ increases.}
\begin{center}
\label{tab:thr}
\begin{tabular}{c c}
\toprule
$\lambda_{d}$ & $t_{opt}$ \\
\midrule
$0.0$ & $0.54$ \\
\midrule
$0.05$ & $0.48$ \\
\midrule
$0.1$ & $0.44$ \\
\midrule
$0.15$ & $0.38$\\
\midrule
$0.2$ & $0.36$ \\
\bottomrule
		
\end{tabular}
\end{center}
\end{table}

\bibliographystyle{plain}
\bibliography{evacref}

\remove{

}

\end{document}